\documentclass[11pt]{article}
\pdfoutput=1
\usepackage[T1]{fontenc}
\usepackage{lmodern}
\usepackage{comment}
\usepackage{slantsc}
\usepackage[protrusion=true,expansion=true]{microtype}
\usepackage{breakcites}
\usepackage{amsmath,amssymb,amsfonts,amsthm, mathtools}
\usepackage{subcaption}
\usepackage{graphicx}
\usepackage{fullpage}
\usepackage{setspace}
\usepackage[backref=page]{hyperref}
\usepackage[nameinlink]{cleveref}
\usepackage{color}
\usepackage{wrapfig}
\usepackage{tikz}
\usetikzlibrary{decorations.pathreplacing,patterns}
\usepackage{algorithm}
\usepackage[noend]{algpseudocode}
\usepackage[framemethod=tikz]{mdframed}
\usepackage{xspace}
\usepackage{pgfplots}
\usepackage{framed}
\usepackage{thmtools}
\usepackage{thm-restate}
\usepackage{tabu}
\usepackage{fancyhdr}
\usepackage{placeins}
\usepackage{float}
\pgfplotsset{compat=1.5}

\newtheorem{theorem}{Theorem}[section]
\newtheorem{corollary}[theorem]{Corollary}
\newtheorem{lemma}[theorem]{Lemma}

\newtheorem{definition}[theorem]{Definition}
\newtheorem{remark}[theorem]{Remark}

\newtheorem{invariant}[theorem]{Invariant}
\newtheorem{observation}[theorem]{Observation}

\newtheorem{assumption}[theorem]{Assumption}

\newenvironment{proofof}[1]{\begin{trivlist} \item {\bf Proof
#1:~~}}
  {\qed\end{trivlist}}

\newcommand{\namedref}[2]{\hyperref[#2]{#1~\ref*{#2}}}

\def \EstLevel    {\mdef{\mathsf{EstLevel}}}

\def \EstDot    {\mdef{\mathsf{EstDot}}}
\def \EmbedEst    {\mdef{\mathsf{EmbedEst}}}

\def \DotEst    {\mdef{\mathsf{DotEst}}}
\def \FPSmall    {\mdef{\mathsf{FPSmall}}}

\def \MaintainList    {\mdef{\mathsf{MaintainList}}}

\def \GetIterate    {\mdef{\mathsf{GetIterate}}}

\newcommand\norm[1]{\left\lVert#1\right\rVert}
\newcommand{\PPr}[1]{\ensuremath{\mathbf{Pr}\left[#1\right]}}

\newcommand{\EEx}[2]{\ensuremath{\underset{#1}{\mathbb{E}}\left[#2\right]}}
\renewcommand{\O}[1]{\ensuremath{\mathcal{O}\left(#1\right)}}

\newcommand{\eps}{\varepsilon}

\def \calA    {\mdef{\mathcal{A}}}

\def \calC    {\mdef{\mathcal{C}}}

\def \calL    {\mdef{\mathcal{L}}}

\def \calP    {\mdef{\mathcal{P}}}

\def \calS    {\mdef{\mathcal{S}}}

\def \calX    {\mdef{\mathcal{X}}}

\def \bA    {\mdef{\mathbf{A}}}
\def \bB    {\mdef{\mathbf{B}}}
\def \bD    {\mdef{\mathbf{D}}}

\def \bU    {\mdef{\mathbf{U}}}

\def \ba    {\mdef{\mathbf{a}}}

\def \be    {\mdef{\mathbf{e}}}

\def \bq    {\mdef{\mathbf{q}}}
\def \bs    {\mdef{\mathbf{s}}}

\def \bu    {\mdef{\mathbf{u}}}
\def \bw    {\mdef{\mathbf{w}}}
\def \bv    {\mdef{\mathbf{v}}}
\def \bx    {\mdef{\mathbf{x}}}
\def \by    {\mdef{\mathbf{y}}}
\def \bz    {\mdef{\mathbf{z}}}

\newcommand{\mdef}[1]{{\ensuremath{#1}}\xspace}  
\DeclareMathOperator*{\Span}{Span}

\DeclareMathOperator*{\argmin}{argmin}

\DeclareMathOperator*{\polylog}{polylog}
\DeclareMathOperator*{\poly}{poly}

\DeclareMathOperator*{\EMD}{EMD}

\DeclareMathOperator*{\Cost}{cost}

\DeclareMathOperator*{\Proj}{Proj}

\newcommand{\ignore}[1]{}

\newif\ifnotes\notestrue 
\ifnotes
\newcommand{\samson}[1]{\textcolor{red}{{\bf (Samson:} {#1}{\bf ) }} \marginpar{\tiny\bf
             \begin{minipage}[t]{0.5in}
               \raggedright S:
            \end{minipage}}}
\newcommand{\elena}[1]{\textcolor{magenta}{{\bf (Elena:} {#1}{\bf ) }} \marginpar{\tiny\bf
    \begin{minipage}[t]{0.5in}
        \raggedright E:
    \end{minipage}}}
\newcommand{\david}[1]{\textcolor{purple}{{\bf (David:} {#1}{\bf ) }} \marginpar{\tiny\bf
             \begin{minipage}[t]{0.5in}
               \raggedright D:
            \end{minipage}}}

\else
\newcommand{\samson}[1]{}
\newcommand{\david}[1]{}
\fi

\makeatletter
\renewcommand*{\@fnsymbol}[1]{\textcolor{mahogany}{\ensuremath{\ifcase#1\or *\or \dagger\or \ddagger\or
 \mathsection\or \triangledown\or \mathparagraph\or \|\or **\or \dagger\dagger
   \or \ddagger\ddagger \else\@ctrerr\fi}}}
\makeatother

\providecommand{\email}[1]{\href{mailto:#1}{\nolinkurl{#1}\xspace}}

\definecolor{mahogany}{rgb}{0.75, 0.25, 0.0}
\definecolor{darkblue}{rgb}{0.0, 0.0, 0.55}
\definecolor{darkpastelgreen}{rgb}{0.01, 0.75, 0.24}
\definecolor{darkgreen}{rgb}{0.0, 0.2, 0.13}
\definecolor{darkgoldenrod}{rgb}{0.72, 0.53, 0.04}
\definecolor{darkred}{rgb}{0.55, 0.0, 0.0}
\definecolor{forestgreenweb}{rgb}{0.13, 0.55, 0.13}
\definecolor{greencss}{rgb}{0.0, 0.5, 0.0}
\definecolor{bleudefrance}{rgb}{0.19, 0.55, 0.91}
\definecolor{darkpastelpurple}{rgb}{0.59, 0.44, 0.84}
\definecolor{darkcerulean}{rgb}{0.03, 0.27, 0.49}

\hypersetup{
     colorlinks   = true,
     citecolor    = mahogany,
	 linkcolor	  = forestgreenweb
}

\fancypagestyle{pg}
{
\lhead{}
\rhead{}
\cfoot{--\ \thepage\ --}

}

\AtBeginDocument{%
  \DeclareFontShape{T1}{lmr}{m}{scit}{<->ssub*lmr/m/scsl}{}%
}
\begin{document}

\allowdisplaybreaks

\title{Adversarial Robustness for Small Frequency Moments and a Weak Equivalence Theorem for Turnstile Streams}

\author{Elena Gribelyuk \\ Princeton University \\ \email{eg5539@princeton.edu}  
\and 
Honghao Lin \\ Carnegie Mellon University \\ \email{honghaol@andrew.cmu.edu} 
\and 
David P. Woodruff \\ Carnegie Mellon University \\ \email{dwoodruf@andrew.cmu.edu}
\and
Huacheng Yu \\ Princeton University \\ \email{hy2@cs.princeton.edu}
\and
Samson Zhou \\ Texas A\&M University \\ \email{samsonzhou@gmail.com}}
\date{\today}

\maketitle{}

\begin{abstract}
We study adversarially robust algorithms for insertion-deletion (turnstile) streams, where future updates may depend on past algorithm outputs. While recent work achieved a robust $(1+\varepsilon)$-approximation for the second moment $F_2$ in polylogarithmic space, achieving high accuracy for other frequency moments remained a major open question; for $p \in [0, 2)$, including the fundamental distinct elements problem ($F_0$), only constant-factor approximations were known in sublinear space. We close this gap, showing that $(1+\varepsilon)$-approximate robustness can be achieved in polylogarithmic space for all $p \in [0, 2]$. Our approach generalizes the estimator-corrector-learner framework to non-Hilbert spaces by dynamically maintaining implicit isometric embeddings into $L_2$ and performing regularized kernel ridge regression over adaptively discovered hard queries, yielding the first insertion-deletion algorithms that approximate: (1) the $p$-th frequency moment $F_p$ up to a $(1+\varepsilon)$-factor in $\text{poly}(1/\varepsilon, \log n)$ space for all $p \in [0, 2]$, including the support size $F_0$, (2) metric and information-theoretic quantities, including the Earth Mover Distance (EMD) and $k$-median clustering cost over $[\Delta]^d$ up to an $\mathcal{O}(d \log \Delta)$-factor, and the Shannon entropy up to an $\varepsilon$-additive error, and (3) non-normed symmetric losses defined by Bernstein functions up to a $(1+\varepsilon)$-factor. For the $F_p$ moments, our algorithm is optimal up to $\text{poly}(1/\varepsilon, \log n)$ factors. Furthermore, we establish a weak equivalence between classical oblivious sketching and adversarial robustness. We prove that for any sub-multiplicative norm, the existence of an efficient classical linear sketch is equivalent to the existence of an efficient adversarially robust turnstile algorithm, up to polynomial factors, formalizing $L_1$ embeddability as the fundamental mechanism governing both models.
\end{abstract}

\pagestyle{pg}
\setcounter{page}{1}

\section{Introduction}
The streaming model of computation has emerged as a central framework for studying algorithms that process massive amounts of sequential data. 
In this setting, an underlying dataset is implicitly defined by a stream of updates and an algorithm must maintain a compact summary to approximately answer queries or report relevant statistics about the data. 
Often, the algorithm is restricted to a single pass over the stream and must use space sublinear in both the stream length $m$ and the universe size $n$. 
This framework has proven remarkably effective for reasoning about large-scale data processing, and has led to a rich body of work spanning algorithm design, lower bounds, and applications.

 In the classical streaming setting, a standard assumption in the literature is that the input stream is fixed in advance and independent of the algorithm’s internal randomness.
However, this assumption is often violated in modern applications, where the input stream may be generated adaptively based on previous outputs of the algorithm.
For instance, queries to a database may depend on previous answers, iterative scientific procedures may update their state using previously computed estimates, and feedback loops may arise in online recommendation platforms and financial systems. 
In such settings, the input stream may depend on the behavior of the algorithm itself (and on the internal randomness), thus raising the question of whether classical streaming guarantees continue to hold.

\paragraph{Adversarially robust streaming.}
To capture this phenomenon, \cite{Ben-EliezerJWY22} introduced the adversarially robust streaming model, where future updates may be chosen adaptively based on previous outputs of the algorithm. 
Specifically, at each time $t \in [m]$, the algorithm receives an update $(a_t, \Delta_t)$ to the underlying frequency vector $\bx \in \mathbb{Z}^n$, where $a_t \in [n]$ is an index and $\Delta_t \in \mathbb{Z}$ denotes an increment or decrement to coordinate $a_t$ of $\bx$ by $\Delta_t$. Concretely, the frequency vector $\bx$ is defined implicitly by $x_i^{(t)} = \sum_{t' \leq t: a_{t'} = i} \Delta_{t'}$. The algorithm must then output an estimate of a function $f(\bx^{(t)})$ of the current prefix. 
Crucially, the algorithm must remain accurate at every step, even when each update $(a_t, \Delta_t)$ is chosen adaptively based on the entire transcript of previous updates and responses of the algorithm. 

\begin{definition}
[\cite{Ben-EliezerJWY22}]
\label{def:model}
Let $f: \mathbb{Z}^n \to\mathbb{R}$ be a fixed function. 
A streaming algorithm $\calA$ is adversarially robust if, for any accuracy parameter $\eps\in(0,1)$, failure parameter $\delta\in(0,1)$, and $m=\poly(n)$, at each time step $t \in [m]$ the algorithm outputs an estimate $Z_t$ such that
\[\PPr{|Z_t - f(\bx^{(t)})|\le\eps\cdot f(\bx^{(t)})}\ge1-\delta.\]
\end{definition}

Informally, the model captures an interaction between a randomized algorithm $\calA$ and an adaptive adversary that constructs a sequence of updates $\{u_1,\ldots,u_m\}$, where each update $u_t$ may depend on all previously observed updates and outputs. 
The interaction proceeds as follows:
\begin{enumerate}
\item 
For each $t \in [m]$, the adversary selects an update $u_t$, where the choice of $u_t$ may be an arbitrary (possibly randomized) function of $u_1,\ldots,u_{t-1}$ and $Z_1,\ldots,Z_{t-1}$.
\item 
Upon receiving $u_t$, the algorithm $\calA$ updates its internal state to incorporate $u_t$.
\item 
The algorithm then outputs an estimate $Z_t$ of $f(\bx^{(t)})$ prior to the next update.
\end{enumerate}
Designing algorithms that remain accurate under such adaptivity, while still using limited space and a single pass, has emerged as a central challenge in recent work~\cite{MironovNS11,MitrovicBNTC17,AvdiukhinMYZ19,Ben-EliezerY20,CherapanamjeriN20,AlonBDMNY21,BravermanHMSSZ21,AjtaiBJSSWZ22,Ben-EliezerJWY22,Ben-EliezerEO22,BeimelKMNSS22,ChakrabartiGS22,MenuhinN22,NaorO22,AssadiCGS23,CherapanamjeriSWZZZ23,CohenNSS23,DinurSWZ23,PengR23,WoodruffZZ23b,AhmadianCohen2024,FengJW24,GribelyukLWYZ24,WoodruffZ24,CohenSS25,GribelyukLWYZ25,AndoniHKO26,Ben-EliezerOS26,CohenGNS26,CohenNSSS26,GalKSY26,GribelyukLWYZ26,LinSWXZ26}. 

\paragraph{Norm and distinct element estimation.}
A canonical class of problems in the streaming model is the estimation of the frequency moments of the underlying data. 
These moments capture a variety of fundamental statistics of the data stream and are formally defined as follows. 
For a frequency vector $\bx \in \mathbb{Z}^n$, the $p$-th frequency moment is given by $F_p(\bx)=\sum_{i\in[n]}|x_i|^p$ for $p>0$, while $F_0(\bx)=|\{i\in[n] : x_i \neq 0\}|$ denotes the number of distinct elements in the stream. 
These quantities capture a range of important properties of the data; for example, $F_1$ corresponds to the total mass of the stream, $F_2$ is closely related to the collision probability, and $F_0$ measures the support size of the underlying vector. 
More broadly, the family of $F_p$ moments for $p>0$ provides a rich spectrum of statistics, and via the identity $\|\bx\|_p = (F_p(\bx))^{1/p}$, directly corresponds to the problem of estimating $L_p$ norms.

These problems have been studied extensively in the classical (non-adaptive) streaming model since the seminal work of \cite{AlonMS99}, with tight space bounds known for a wide range of parameters~\cite{AlonMS99, Indyk06,Li08,KaneNW10a,KaneNW10b,KaneNPW11}. 
For the distinct elements problem corresponding to $p =0$, it is known that $\tilde{\Theta}\left(\frac{1}{\eps^2}\log n\right)$ bits of space is both necessary~\cite{AlonMS99,KaneNW10a} and sufficient~\cite{KaneNW10b}. 
Similarly for $p\in(0,2]$, it is known that for all $\eps\in(n^{-1/2+c},1)$ for any constant $c>0$, $\Theta\left(\frac{1}{\eps^2}\log n\right)$ bits of space is both necessary~\cite{AlonMS99,KaneNW10a} and sufficient~\cite{AlonMS99,KaneNW10a,KaneNPW11}. 
Finally, for $p>2$, it is known that $\tilde{\Theta}\left(\frac{1}{\eps^2}\cdot n^{1-2/p}\right)$ bits of space is both necessary~\cite{Ganguly12,WoodruffZ21b} and sufficient~\cite{Ganguly11,GangulyW18}. 
Similar bounds are known for $p\in[0,2]$ and integer $p>2$ in the adversarially robust model when the adversary is restricted to $\Delta_t\ge 0$ for all $t\ge 0$, i.e., the insertion-only setting~\cite{WoodruffZ21}. 

In contrast, significantly less is understood in the adversarially robust setting in the presence of deletions. 
A long line of work has shown that linear sketching, which underlies essentially all known turnstile streaming algorithms in the non-adaptive setting~\cite{LiNW14,AiHLW16,HosseiniLY19,KallaugherP20,JiangLY26}, is fundamentally limited in this setting, and cannot achieve even a constant-factor approximation to these quantities unless the sketching dimension is $r=\Omega(n)$, i.e., the sketch requires near-linear space~\cite{HardtW13,GribelyukLWYZ24,GribelyukLWYZ25}. 
This has led to a significant gap in our understanding of the power and limitations of turnstile streaming algorithms under adaptive updates.

Recent work has begun to close this gap by demonstrating that 
it is possible to design sublinear-space adversarially robust algorithms in the turnstile model. 
In particular, it is now known that a $(1+\eps)$-approximation for the $L_2$ norm can be achieved using $\poly\left(\frac{1}{\eps},\log n\right)$ space~\cite{GribelyukLWYZ26}. 
For $p\neq 2$ and for the distinct elements problem, however, only constant-factor approximations were previously known in sublinear space, leaving a significant gap in our understanding of the achievable approximation guarantees under adversarial updates.

\subsection{Our Contributions}
In this work, we significantly advance the understanding of adversarial robustness in the insertion-deletion streaming model. 
Our main results establish that, despite strong impossibility results for linear sketches~\cite{HardtW13,GribelyukLWYZ24,GribelyukLWYZ25}, accurate estimation of a broad class of frequency moments is still achievable in sublinear space under adaptive updates. 
We further develop a conceptual connection between adversarial robustness and classical sketching, showing that these two models are, in a weak sense, equivalent up to polynomial losses. 
We also show that our robust $F_p$ estimation primitives can be combined with standard metric and norm embeddings to obtain the first adversarially robust sublinear-space algorithms for a variety of other problems. 
These include geometric and clustering objectives such as Earth Mover Distance and $k$-median cost, as well as statistical measures like Shannon entropy. 

\paragraph{$F_p$ moment estimation for $p \in [0,2]$.}
Our first main result resolves the approximability of frequency moments and distinct elements in the adversarially robust turnstile model for the entire regime $p \in [0,2]$. 
Prior to our work, $(1+\eps)$-approximation was only known for $p=2$, while for all other values of $p \in [0,2)$ (including the distinct elements problem $F_0$), only constant-factor approximations were known in sublinear space~\cite{GribelyukLWYZ26}.

\begin{restatable}[Robust $F_p$ moment estimation]{theorem}{thmlpsmall}
\label{thm:lp:small}
Given a constant $p\in[0,2]$ and any $\eps\in(0,1)$, there exists an adversarially robust insertion-deletion streaming algorithm on a stream of length $m$ that with high probability, outputs a $(1+\eps)$-approximation to the $F_p$ moment at all times, for the underlying frequency vector of universe size $n$. 
For $m=\poly(n)$, the algorithm uses $\poly\left(\frac{1}{\eps},\log n\right)$ bits of space. 
\end{restatable}

\Cref{thm:lp:small} provides a unified and nearly optimal guarantee for all $p \in [0,2]$, matching the space complexity of classical (non-adaptive) streaming up to polynomial factors~\cite{AlonMS99,KaneNW10a,KaneNW10b,KaneNPW11}. 
In particular, \Cref{thm:lp:small} gives the first $(1+\eps)$-approximation for the distinct elements problem in the adversarially robust turnstile model using polylogarithmic space, thereby closing a central open question in the area.

\paragraph{Equivalence between robustness and sketching.}
Our second main result provides a structural explanation for the power and limitations of adversarially robust streaming algorithms. 
As previously discussed, linear sketching is known to be inherently non-robust to adaptive inputs, yet it underlies essentially all optimal oblivious streaming algorithms. 
We show that this tension can be formalized via a general equivalence theorem.

A key conceptual ingredient underlying this result is again the central role of $L_1$. 
Building on \Cref{thm:lp:small}, which gives a robust and accurate estimator for the $L_1$ norm, we leverage a foundational connection between sketching and metric embeddings. 
Specifically, a seminal result of \cite{AndoniKR18} establishes that for finite-dimensional normed spaces, admitting an efficient classical linear sketch is mathematically equivalent to admitting a low-distortion embedding into $L_{1-\eps}$ (and into $L_1$ for norms closed under sum-product). 
By this characterization, \emph{any} normed space that admits an efficient oblivious sketch must embed into such spaces. 
This allows us to simulate arbitrary linear sketches in the adversarially robust setting by composing these guaranteed embeddings with our robust $L_p$ estimation primitives. 
Conversely, any adversarially robust algorithm can be converted into a classical sketch by compressing its interaction transcript. 
Together, these reductions yield a general equivalence between the two models.

\begin{restatable}[Weak equivalence]{theorem}{thmequivalence}
\label{thm:equivalence}
Let $\calX = (\mathbb{R}^n, \|\cdot\|_{\calX})$ be a finite-dimensional normed space. 
Then:
\begin{enumerate}
\item 
\textbf{Forward direction (oblivious to robust):} 
If $\calX$ admits a linear sketch of size $s$ bits solving the $D$-gap norm problem, then $\calX$ admits an adversarially robust turnstile streaming algorithm in the advice model that achieves an approximation factor of $\O{sD}$ while using $\poly(s, \log n)$ bits of space. 
\item 
\textbf{Reverse direction (robust to oblivious):} 
Conversely, suppose $\calX$ admits an adversarially robust turnstile streaming algorithm that uses $s$ bits of space and achieves a multiplicative $D$-approximation for sufficiently large data streams.  
Then there exists a linear sketch for $\calX$ in the classical oblivious setting that uses $\O{s\log n}$ bits of space. 
\end{enumerate}
\end{restatable}

\Cref{thm:equivalence} establishes a weak equivalence between the adversarially robust and oblivious models: any progress in one model translates moderately to the other. 
In particular, \Cref{thm:equivalence} shows that the existence of efficient adversarially robust algorithms implies new linear sketches, while classical sketching lower bounds extend to the robust setting. 
Thus, this connection provides a unifying framework for understanding which problems admit efficient robust algorithms, and explains the inherent difficulty of extending linear sketches to the adversarial setting. 

A central theme underlying both our algorithms and our structural results is the role of $L_1$ as a unifying geometry. 
At a technical level, our $(1+\eps)$-approximation for $F_p$ in \Cref{thm:lp:small} is obtained by reducing to robust $F_1$ estimation via suitable transformations. 
More broadly, rather than $L_1$ merely serving as a convenient target space for a wide range of embeddings arising in theoretical computer science, the equivalence of~\cite{AndoniKR18} demonstrates that $L_1$ (and $L_{1-\eps}$) embeddability is the fundamental and unavoidable mechanism for sketching norms. 
This perspective is powerfully reinforced by \Cref{thm:equivalence}, whose proof proceeds by extracting such embeddings from classical sketches and simulating them in the adversarially robust setting. 
Taken together, these results establish a compelling triad---adversarial robustness, oblivious sketching, and $L_1$ embeddability are intrinsically tied together---and suggest that robust $L_1$ estimation is a fundamental primitive from which more general adversarially robust streaming algorithms can be derived.

\paragraph{Additional applications.}
Building on this viewpoint, we obtain a number of additional applications by combining known embeddings into $L_p$ with our robust estimation framework. 
We highlight a few representative examples below.

Our first application is to the Earth Mover Distance (EMD), a central metric for comparing distributions over discrete domains. 
It is well-known that EMD over the grid $[\Delta]^d$ admits an embedding into $L_1$ with distortion $\O{d \log \Delta}$. 
By combining this embedding with our robust $F_1$ estimation primitives, we obtain an adversarially robust approximation algorithm for EMD in the turnstile model.

\begin{restatable}[Robust Earth Mover Distance]{theorem}{thmemd}
\label{thm:emd}
There exists an adversarially robust turnstile streaming algorithm of length $m$ over universe $[\Delta]^d$ that computes a $\O{d\log\Delta}$-approximation to the Earth Mover Distance $\|\bx\|_{\EMD}$ over the grid $[\Delta]^d$, using $\poly(d,\log(m\Delta))$ bits of space. 
\end{restatable}

Our second application is to the $k$-median clustering objective, which can also be reduced to $L_1$ via standard embeddings. 
Using our framework, we obtain adversarially robust algorithms for maintaining approximate clustering solutions on insertion-deletion streams. 

\begin{restatable}[Robust $k$-Median Clustering]{theorem}{thmkmedian}
\label{thm:kmedian}
There exists an adversarially robust turnstile streaming algorithm that outputs a set of $k$ centers that provide a $\O{d\log\Delta}$-approximation to the $k$-median clustering cost over $[\Delta]^d$, using $\poly(k,d,\log\Delta)$ bits of space.
\end{restatable}

Finally, we consider entropy estimation, which can be reduced to estimating frequency moments. 
Using our robust $F_p$ estimators as a black-box primitive, we obtain accurate entropy estimates even in the presence of adversarial updates.

\begin{restatable}[Robust entropy estimation]{theorem}{thmentropy}
Given any $\eps\in(0,1)$, there exists an adversarially robust insertion-deletion streaming algorithm on a stream of length $m$ that with high probability, outputs a $\eps$-additive approximation to the Shannon entropy, for the underlying frequency vector of universe size $n$. 
For $m=\poly(n)$, the algorithm uses $\poly\left(\frac{1}{\eps},\log n\right)$ bits of space. 
\end{restatable}

Beyond these examples, our framework applies more broadly to any problem that admits a low-distortion embedding into $L_1$ or can be reduced to frequency moment estimation. 
This includes a variety of geometric, statistical, and combinatorial estimation tasks, highlighting the generality of our approach and the central role of $L_p$ embeddings in the adversarially-robust streaming model.

\subsection{Technical Overview}

\subsubsection{\texorpdfstring{$F_p$}{Fp} Estimation Algorithm for \texorpdfstring{$p<2$}{p<2}}
A natural starting point for \Cref{thm:lp:small} might be to generalize the existing adversarially robust $(1+\eps)$-approximate $L_2$ estimation algorithm that handles insertions and deletions. 

\paragraph{Estimator-corrector-learner framework.}
At a high level, the $F_2$ algorithm of \cite{GribelyukLWYZ26} departs from classical linear sketching by introducing fresh randomness during the execution of the algorithm and organizing the computation around a estimator-corrector-learner paradigm. 
The stream is conceptually partitioned into blocks, and within each block we decompose the underlying frequency vector as $\bx=\bz-\bq$, where $\bz$ corresponds to the portion of the stream before the block. 
For convenience, we write $-\bq$ as the updates of the block, so that $\|\bx\|$ can be interpreted as a metric measuring the distance between $\bz$ and $\bq$. 

The algorithm maintains multiple sketches that correspond to three ingredients, which play different roles in this decomposition:
\begin{enumerate}
\item
The \emph{estimator} produces a candidate approximation to $\|\bz-\bq\|_2^2$ via $\|\bz-\bz'\|_2^2+\|\bz'-\bq\|_2^2$, where $\bz'$ is a current guess for the vector $\bz$, which arrived before the current block and is thus unknown from the perspective of the sketches initialized at the beginning of the block. 
\item
The algorithm also maintains a \emph{corrector}, which uses a separate sketch on $\bz-\bq$ initialized at the beginning of the data stream, to verify whether the estimator is accurate.  
If the estimate is correct, the algorithm proceeds using the output of the estimator; otherwise, the corrector overrides the estimator and outputs an estimate for $\|\bz - \bq\|_2^2$, and activates the learner. 
\item
Upon detecting an incorrect estimate, the \emph{learner} updates the internal iterate $\bz'$, refining its approximation to $\bz$ using information revealed by the current query $\bq$, for which the estimator was not accurate.
\end{enumerate}
The key idea is that each time that the estimator is not accurate, this exposes structural information about $\bz$ to the algorithm, allowing us to update the guess $\bz'$ to move closer to $\bz$ in a controlled manner. 
After the update, subsequent estimators operate with the refined iterate, and the process continues.

A crucial property of the framework is that the number of learner updates is small. 
Each query $\bq$ can only correspond to an incorrect output by the estimator if $\bq$ is aligned with some direction of $\bz$ not captured by the guess $\bz'$. 
For example, when the estimator errs via $\|\bz-\bz'\|_2^2+\|\bz'-\bq\|_2^2>(1+\eps)\cdot\|\bz-\bq\|_2^2$, this implies
\begin{align*}
    \langle\bz-\bz',\bz'-\bq\rangle\ge\eps\cdot\|\bz-\bz'\|_2\cdot\|\bz'-\bq\|_2.
\end{align*}
Thus, it can be shown that updating $\bz'$ toward $\bq$ by setting $\bz''=\bz'+\alpha(\bz'-\bq)$ for some carefully chosen $\alpha\approx\frac{\eps\cdot\|\bz-\bz'\|_2}{\|\bz'-\bq\|_2}$ makes quantifiable progress toward $\bz$:
\begin{align}\label{eq:progress:L2}
    \|\bz-\bz''\|_2^2\le(1-\eps^2)\cdot\|\bz-\bz'\|_2^2.
\end{align}
Since the frequency vector has entries bounded by $\poly(n)$, it follows that the iterate $\bz'$ converges to $\bz$ after only $\O{\frac{1}{\eps^2} \log n}$ updates to $\bz'$.
Therefore, the corrector sketch only needs to be accurate for a small number of adaptive queries, which can be ensured with only polylogarithmic overhead in the total space complexity. Then, the final insight is that robustly estimating the second term $\|\bz' - \bq\|_2^2$ is a problem of precisely the same form as the original problem, and thus it can be recursively computed by implementing the estimator-corrector-learner framework with estimator $\|(\bz' - \bq_1) - \bu\|_2^2 + \|\bu - \bq_2\|_2^2$, where $\bq = \bq_1 + \bq_2$ is decomposed so that segments $\bq_1$ and $\bq_2$ can be sketched using independent sketches.  

\paragraph{An alternative update rule.}
We remark that while \cite{GribelyukLWYZ26} analyzed the algorithm for the update rule $\bz'' \gets \bz' + \alpha(\bz' - \bq)$, one can show the same convergence as in \Cref{eq:progress:L2} for a modified update rule, given as follows: let $\calL = \{\bq_i\}_{i = 1}^L$ be a list of queries on which the estimator was previously inaccurate. If the estimator is not accurate on the current suffix $\bq$, we update 
\[\bz'' := \Proj_{\Span(\calL \cup \{\bq\})} (\bz),\] 
noting that the projection can be computed by solving a regression problem in the sketch space. In fact, this update rule has the nice (but not strictly necessary) property that once the iterate is updated to $\bz''$, the new estimator $\|\bz - \bz''\|_2^2 + \|\bz'' - \bq\|_2^2$ is precisely equal to $\|\bz - \bq\|_2^2$ by Pythagorean theorem.

\paragraph{Possible approaches for $F_p$ estimation.}
In order to generalize the framework of \cite{GribelyukLWYZ26} for other $L_p$ norms, we must determine an appropriate decomposition of $\|\bx\|_p^p = \|\bz - \bq\|_p^p \approx g(\bz-\bz') + g(\bz' - \bq)$ for some sketchable function $g: \mathbb{R}^n \rightarrow \mathbb{R}$ with the following properties: 

\begin{itemize}
    \item The prefix $\bz- \bz'$ and suffix $\bz' - \bq$ can be sketched independently. Moreover, for $\bq = \bq_1 + \bq_2$, the function $g$ should have the property that a sketch for $g(\bz)$ can subsequently be used to track $g(\bz- \bq_1)$. Note that this is needed in order to implement the recursion, as described above. 
    \item This decomposition should allow us to argue that if $g(\bz-\bz') + g(\bz' -\bq) \not \in (1\pm\eps) \cdot \|\bz - \bq\|_p^p $, the algorithm can recover more information about $\bz$ from the current query $\bq$.
\end{itemize} 

Recall that in the $L_2$ algorithm of \cite{GribelyukLWYZ26}, the crucial insight was that due to the Pythagorean properties of Euclidean space and the squared $L_2$ norm, an inaccurate estimator directly implies that the dot product between $\bz-\bz'$ and $\bz'-\bq$ must be significant, thus enabling the algorithm to make progress in learning $\bz$. Unfortunately, it is not clear how to establish a similar claim for $L_p$ norms. 
For instance, one might try to write down the binomial expansion of $\|\bz - \bq\|_p^p$ for integer $p$, but then observe that incorrect outputs by the estimator might be caused by numerous correlations between coordinate-wise powers of $\bz-\bz'$, $\bz'-\bq$, and $\bz-\bq$. 
Thus, even by considering a more complex potential function of $\bz-\bz'$, it is not at all clear how to show convergence. 

\paragraph{A hard instance: the Hadamard code.}
Our first observation is that the iterate $\bz'$ maintained by the learner of \cite{GribelyukLWYZ26} uses a very specific update rule $\bz'\gets\bz'+\alpha(\bz'-\bq)$ and so $\bz'$ can be written as a linear combination of a list $\bu_1,\ldots,\bu_L$ of previous queries on which the estimator failed. 
Hence, $\bz'$ is contained within the subspace $\bU$ spanned by the vectors $\bu_1,\ldots,\bu_L$. 
Thus, the distance $\|\bz'-\bz\|_2^2$ between the iterate $\bz'$ and the target vector $\bz$ is at least the distance between $\bz$ and the closest vector $\bv$ to $\bz$ contained within the subspace $\bU$, i.e., $\bv=\argmin_{\by\in\mathbb{R}^L}\|\bU\by-\bz\|_2^2$. 
Therefore, we could similarly upper bound the number of iterations required for convergence if at each step, we simply take the closest vector $\bv$ to $\bz$ contained within the subspace $\bU$. 

We could try adapting this approach to $L_p$ norms by maintaining the list of queries $\bu_1,\ldots,\bu_L$ on which the estimator failed and iteratively updating $\bz'$ to be the closest vector to the subspace $\bU$ in $L_p$ distance. 
However, it turns out that this approach fails catastrophically. 
Suppose $\bz=\mathbf{1}^n$ and each query $\bq\in\{-1,+1\}^n$ is a codeword in the Hadamard code not equal to $\bz$, so that $\sum_{i=1}^n q_i=0$. 
Hence, once such a query is added, the maintained subspace $\bU=\Span(\bu_1,\ldots,\bu_L)$ is contained in $\left\{\by \in\mathbb{R}^n:\sum_{i=1}^n y_i=0\right\}$.
The update step sets $\bz'$ to be the minimizer of $\min_{\by \in \bU} \|\mathbf{1}^n - \by\|_p^p$. 
For $p>1$, the function $t \mapsto |1-t|^p$ is strictly convex and so by Jensen's inequality, for any $\by \in \bU$,
\[\frac{1}{n}\sum_{i=1}^n |1 - y_i|^p \ge \left|1 - \frac{1}{n}\sum_{i=1}^n y_i\right|^p = 1,\]
where we used that $\sum_i y_i = 0$. 
Moreover, equality holds if and only if all coordinates of $\by$ are equal, which together with the zero-sum constraint forces $\by = \mathbf{0}^n$. 
Thus the unique minimizer is $\bz' = \mathbf{0}^n$.
Note that while $\mathbf{0}^n$ is also the minimizer for $p=2$, the algorithm avoids this trap because the initial estimator evaluates to $\|\bz-\bz'\|_2^2 + \|\bz'-\bq\|_2^2 = n+n=2n$, which matches the true distance $\|\bz-\bq\|_2^2 = \frac{n}{2}|1-1|^2 + \frac{n}{2}|1-(-1)|^2 = 2n$, meaning the corrector never flags these queries. 
For $p \in (1, 2)$, however, the true distance is $\|\bz-\bq\|_p^p = n 2^{p-1}$, which is a constant-factor multiplicative separation from the estimate of $2n$.
Consequently, $\bz'$ remains $\mathbf{0}^n$ throughout the process, so that it never progresses toward $\bz$, even while each query is continuously flagged by the corrector and added to the list.

\paragraph{Embedding into $L_2$ and regularized regression.}
We first recall the existence of isometric embeddings $f$ onto Hilbert spaces~\cite{schoenberg1938metric,schoenberg1935definition,robertson2024negative}, so that for any $p\in[0,2)$, we have for all vectors $\bu,\bv\in\mathbb{R}^n$,
\[\|\bu-\bv\|_p^p=\|f(\bu)-f(\bv)\|_2^2,\qquad \|f(\mathbf{0})\|_2=0,\]
so that for the $p$-th moment of the vector $\bz-\bq$, it suffices to estimate $\|f(\bz)-f(\bq)\|_2^2$. In the discussion below, we let $N$ denote the embedding dimension of $f(\cdot)$. 
Ideally, if we could apply the embedding on the fly to each arriving element in the data stream, then this problem would directly reduce to $F_2$ estimation.  Unfortunately, these embeddings are non-linear mappings into a higher dimensional space and it is not clear how to maintain embedded vectors in the streaming setting.

In fact, it turns out that we only require the \textit{existence} of such embeddings to design our robust $F_p$ estimation algorithm by simulating the estimator-corrector-learner paradigm \textit{implicitly} in the embedded space. As such, we now describe the algorithm as though we have access to embedded vectors, although the true algorithm will strictly operate on the original vectors.
By analogy to the $L_2$ algorithm, suppose the estimator for $\|\bz - \bq\|_p^p = \|f(\bz) - f(\bq)\|_2^2$ is $\|f(\bz) - \bv \|_2^2 + \|\bv - f(\bq)\|_2^2$, where $\bv = \Proj_{\calL} f(\bz)$ is the iterate that approximates $f(\bz)$ and $\mathcal{L} = \{f(\bq_i)\}$ is the list of previous queries on which the estimator was not accurate \footnote{As we explain next, $\bv$ can be computed via a kernel regression over these previously observed queries $f(\bq_1),f(\bq_2),\ldots,f(\bq_L)$ on which the estimator failed. Crucially, this regression can be implemented using only estimates of norms and pairwise inner products of the form $\langle f(\bq_i), f(\bq_j)\rangle$ and $\langle f(\bz), f(\bq_i)\rangle$, which can themselves be obtained using $F_p$ sketches.}. Let $|\calL| = L$. Since $\bv$ is a projection onto the span of the previous hard queries $f(\bq_i)$, we have that $\bv = \sum_{i = 1}^{L} \alpha_i \cdot f(\bq_i)$ for some coefficients $\alpha_i \in \mathbb{R}$. Importantly, note that $\bv$ does not correspond to an explicit vector in the original space, and may not admit a pre-image under the embedding $f$.

Now, we can estimate the quantity $\|f(\bz) - \bv\|_2^2$ in the streaming setting by first observing that 
\begin{align}\label{eq:expansion:first:term}
\|f(\bz) - \bv\|_2^2 = \|f(\bz)\|_2^2 - 2\sum_{i = 1}^{L} \alpha_i \cdot \langle f(\bz), f(\bq_i) \rangle + \sum_{1 \leq i,j \leq L} \alpha_i \cdot \alpha_j \cdot \langle f(\bq_i), f(\bq_j) \rangle.
\end{align}

Then, we can simply estimate all of these right-hand side terms as follows:

\begin{itemize} 
\item $\|f(\bz)\|_2^2 = \|\bz\|_p^p$ can be estimated using an $F_p$ estimation sketch for $\bz$. 
\item $\langle f(\bz), f(\bq_i) \rangle$ can be estimated by noting that 
\[\langle f(\bz), f(\bq_i) \rangle = \frac{1}{2} \left(\|f(\bz)\|_2^2 + \|f(\bq_i)\|_2^2 - \|f(\bz) - f(\bq_i) \|_2^2 \right).\]

Therefore, we can estimate $\|f(\bq_i)\|_2^2 = \|\bq_i\|_p^p$ and $\|f(\bz) - f(\bq_i)\|_2^2 = \|\bz - \bq_i\|_p^p$ using $F_p$ estimation sketches on $\bq_i$ and $\bz - \bq_i$ respectively, so we obtain an estimate of $\langle f(\bz), f(\bq_i) \rangle$ up to $\eps \cdot \left(\|\bz\|_p^p + \|\bq_i\|_p^p + \|\bz-\bq_i\|_p^p \right)$ additive error. 
We remark that the estimate of $\|\bz - \bq_i\|_p^p$ precisely corresponds to a previous output of our algorithm. 
\item By the same reasoning, terms $\langle f(\bq_i), f(\bq_j) \rangle$ can be estimated using $F_p$ estimation sketches for previous (fixed) queries $\bq_i$, $\bq_j$, and $\bq_i - \bq_j$. 
\end{itemize}

As before, the corrector will use an $F_p$ estimation sketch for $\bz-\bq$ to check whether the current estimator is accurate, and to trigger updates to the iterate $\bv$. 

\paragraph{Recursion.}
At this point, it suffices to show how to recursively estimate $\|\bv-f(\bq)\|_2^2$ in the next level of the algorithm. However, since the iterate $\bv$ does not necessarily have a pre-image under $f$, this significantly complicates the implementation of the recursion. To this end, we let $\bq = \bq_L + \bq_{L+1}$ and expand this quantity as follows:
\begin{align}\label{eq:recursion:expansion}
\|\bv - f(\bq)\|_2^2 &= \left \| \sum_{i = 1}^{L}\alpha_i \cdot f(\bq_i) - f(\bq) \right\|_2^2 \\ 
&= \sum_{1 \leq i\leq j \leq L} \left( \alpha_i \cdot \alpha_j \cdot \langle f(\bq_i), f(\bq_j) \rangle \right) - 2 \sum_{i= 1}^{L} \alpha_i \cdot \langle f(\bq_i), f(\bq) \rangle + \|f(\bq)\|_2^2\nonumber
\end{align}

We describe how to compute each of the terms above. 
First, note that since past queries $\bq_i \in \calL$ are fixed, inner products of the form $\langle f(\bq_i), f(\bq_j) \rangle$ can be computed robustly by simply invoking $F_p$ estimation sketches on $\bq_i$, $\bq_j$, and $\bq_i - \bq_j$ (as described earlier). 
Thus, we focus on estimating terms which involve the evolving portion of the stream $\bq = \bq_L + \bq_{L+1}$. 
In particular, by the properties of the embedding $f$, the inner products $\langle f(\bq_i), f(\bq) \rangle $ appearing in the second term can be expressed as

\begin{align*}
\langle f(\bq_i), f(\bq) \rangle &= \frac{1}{2}\left(\|f(\bq_i)\|_2^2 + \|f(\bq)\|_2^2- \|f(\bq_i) - f(\bq)\|_2^2 \right)  \\ 
&= \frac{1}{2}\left(\|f(\bq_i)\|_2^2 + \|\bq_L + \bq_{L+1}\|_p^p - \| \bq_i - (\bq_L + \bq_{L+1})\|_p^p \right) \\ 
& = \frac{1}{2} \left(\|f(\bq_i)\|_2^2 + \|f(\bq_L) - f(-\bq_{L+1})\|_2^2 - \|f(\bq_i-\bq_L) - f(\bq_{L+1})\|_2^2 \right).
\end{align*}

Likewise, the third term can be written as $\|f(\bq)\|_2^2 = \|\bq_L + \bq_{L+1}\|_p^p = \|f(\bq_L) - f(-\bq_{L+1})\|_2^2$. 
At this point, note that estimating each of the above norms is a problem of precisely the same form as the original problem, and it may seem as though we must then recursively estimate each of these $\O{L^2}$ norms arising from inner products in \Cref{eq:recursion:expansion}. 
In fact, it turns out that with some careful algebraic manipulation, we can reduce the problem to only $\textit{one}$ recursive call on the squared-norm $\left \|\sum_{i = 1}^L \alpha_i \cdot f(\bq_i - \bq_L) - f(\bq_{L+1})\right \|_2^2$, and all remaining terms can be robustly computed by directly using $F_p$ estimation sketches on previous queries $\bq_i - \bq_{L}$ or $\bq_i$ for $i \in [L]$. \footnote{We refer the reader to \Cref{lem:recursion:computation} for more details regarding how to estimate various terms in the expansion.}. 

\paragraph{Bounding the additive error.}

Let us examine the total error incurred by our estimates in the approach described above. 
Recall that for any $\bw, \by \in \mathbb{R}^n$ and any $\eps' > 0$, each inner product $\langle f(\bw), f(\by) \rangle$ can be estimated up to error $\eps' \cdot \left(\|\bw\|_p^p + \|\by\|_p^p + \|\bw - \by\|_p^p \right)$. 
Now, suppose that the underlying frequency vector $\bx = \bz - \bq$ satisfies $\|\bx\|_p^p \leq 10 A$ at all times for some fixed parameter $A$ \footnote{Note that for this condition, we require that any any time, the prefix of the stream has $F_p$ moment bounded by $A$. Specifically, $\|\bz\|_p^p \leq \O{A}$ and $\|\bz - \bq_i\|_p^p \leq \O{A}$ for all previous queries $\bq_i$. Also, this directly implies that $\|\bq_i\|_p^p \leq \O{A}$, so $\|\bq_i - \bq_j\|_p^p \leq \O{A}$ as well.}. 
Furthermore, suppose that the coefficients $\alpha_i$ appearing in \Cref{eq:expansion:first:term} and \Cref{eq:recursion:expansion} are bounded by $\O{1}$. 
Then, by setting $\eps' = \frac{\eps}{L^2}$, this means that the error of estimating the first term $\|f(\bz) - \bv\|_2^2$ via dot products in \Cref{eq:expansion:first:term} is roughly $L^2 \cdot \O{\frac{\eps}{L^2} \cdot A} \approx \O{\eps} \cdot A$; by the same argument, we can also ensure that the error due to the second term $\|\bv - f(\bq)\|_2^2$ is $\O{\eps} \cdot A$ as well. 

However, this analysis relied on the assumption that the projection of $f(\bz)$ onto $\Span(\calL)$, given by $\bv = \sum_{i = 1}^L \alpha_i \cdot f(\bq_i)$, has coefficients $\alpha_i$ bounded by $\O{1}$. 
In fact, it turns out that this is not necessarily true for arbitrary list vectors $\bq_1,\ldots, \bq_{L}$. To overcome this, we instead imagine that the list vectors are actually $\bu_i=f(\bq_i)+\nu\cdot\be_{N+i} \in \mathbb{R}^{N + L}$, and let $\calL' = \{\bu_i\}_{i \in [L]}$. 
By scaling, suppose that $\|\bu_i\|_2 = 1$ for all $i \in [L]$, and let $\bU$ be the matrix whose columns are $\bu_1,\ldots, \bu_{L}$. 
First, we observe that $\sigma_{\max}(\bU) \leq \|\bU\|_F \leq \sqrt{L} \cdot \sqrt{\O{1} + \nu^2} = \O{\sqrt{L}}$. 
Moreover, since the $L \times L$ identity matrix scaled by $\nu$ is appended in the dimensions after $N$, we have that $\sigma_{\min}(\bU) \geq \nu$. 
Thus, it follows that the condition number of $\bU$ is $\kappa = \frac{\sqrt{L}}{\nu} = \frac{L \cdot H^2}{\eps^2}$, and thus the projection of $f(\bz)$ (with $0$'s in additional dimensions) onto $\Span(\bu_1,\ldots, \bu_{L})$ will have coefficients satisfying $|\alpha_i| \leq \kappa = \frac{\sqrt{L}}{\nu}$. 

With the above observation in mind, we are now ready to define our modified update rule. 
Note that finding the projection of $f(\bz)$ onto the space spanned by $\bu_1, \ldots, \bu_L$ is equivalent to computing $\argmin_{\bu'\in\Span(\calL')}\|f(\bz)-\bu'\|_2^2$. 
However, since we appended the $L \times L$ identity scaled by $\nu$ in the added dimensions, this actually corresponds to solving the \textit{regularized} regression $\argmin_{\bu' \in \Span(\calL)} \left(\|f(\bz) - \bu' \|_2^2 + \nu^2 \cdot \|\mathbf{\alpha}\|_2^2\right)$, where $\alpha \in \mathbb{R}^L$ denotes the coefficient vector of $\bu'$ in the basis $\calL = \{f(\bq_i)\}_{i \in [L]}$. 
As before, all computations in this regularized regression take place in the sketch space by estimating dot products of the form $\langle f(\bz), f(\bq_i)\rangle $ or $\langle f(\bq_i), f(\bq_j)\rangle$ up to additive error $\O{\eps} \cdot A$. Finally, if $\bu' = \sum_{ i =1}^L \alpha_i \cdot \bu_i$ for $\bu_i \in \calL'$, we output the coefficients of $\bu_1 + \bu'$, which we use to represent the linear combination of our new iterate $\bv$. 

Overall, this ensures that the coefficients $\alpha_i$ in the linear combinations of \Cref{eq:expansion:first:term} and \Cref{eq:recursion:expansion} are bounded, and thus we ensure that the total error accrued in estimating all $\O{L^2}$ dot products is $\O{\eps} \cdot A$ by setting $\eps' = \frac{\eps}{\kappa \cdot L^2}$ to be the error parameter for each sketch.

\paragraph{From additive error to multiplicative error.}
We now describe how to convert a robust additive error algorithm $\calA$ to a robust multiplicative error algorithm $\calA'$ by stratifying the estimation process across multiple scales.

Recall that the algorithm above outputs estimates with error $\eps \cdot A$ when the underlying frequency vector $\bx$ satisfies $\|\bx\|_p^p \leq 10 A$. 
To obtain multiplicative error, we initialize $K = \O{\frac{1}{\eps}\log n}$ independent copies $\calA_1,\ldots, \calA_K$ of the robust additive-error algorithm, where $\calA_i$ achieves error $\eps \cdot (1+\eps)^i$ when $\|\bx\|_p^p \leq 10 \cdot (1+\eps)^i$. 
If we could maintain the property that input vectors to the \textit{robust} algorithm $\calA_i$ for scale $S_i = [(1+\eps)^i, (1+\eps)^{i+1})$ satisfies $\|\bx\|_p^p = \Theta{((1+\eps)^i)}$, then at scale $S_i$, the additive error of the estimator is $\eps \cdot \Theta((1+\eps)^i)$, which corresponds to a multiplicative $(1+\eps)$-approximation for all queries for which the $F_p$ moment lies in $\Theta((1+\eps)^i)$. 
To determine the appropriate scale for a given query $\bz-\bq$, we begin at the top level $K$. 
The estimator at this level provides an additive approximation of $\eps\cdot\Theta((1+\eps)^K)$, which suffices to determine whether the true value is smaller than $2\eps\cdot\Theta((1+\eps)^K)$. 
If this is the case, we can safely proceed to a lower level and repeat the process. 
A more fine-grained version of this iterative search allows the algorithm to iteratively zoom in on the correct level, ultimately yielding a multiplicative $(1+\eps)$-approximation across all scales of $\|\bz-\bq\|_p^p$. 
So, if each of the algorithms $\calA_i$ is accurate with high probability, $\calA'$ returns a multiplicative $(1+\eps)$ approximation. 

Finally, it remains to show that this new algorithm $\calA'$ is robust, i.e. each of the $\calA_i$'s continue to be accurate even when updates are adaptive. 
By contradiction, suppose that there exists a strategy $\calS$ for the adversary to break a particular $\calA_i$ with non-trivial probability. 
In that case, we can construct an adaptive attack on $\calA_i$ by letting the adversary simulate the execution of the algorithm for the remaining scales $\calA_1,\ldots, \calA_{i-1}, \calA_{i+1}, \ldots, \calA_K$ by itself, and then execute strategy $\calS$ to break $\calA_i$. 
However, this contradicts the robustness of $\calA_i$, so we conclude that $\calA'$ is adversarially robust on streams of length $m = \poly(n)$.

\subsubsection{Sketching, EMD, and Clustering}

\paragraph{Equivalence between robustness and sketching.}
To formalize the connection between adversarial robustness and classical sketching, we use a well-known result by \cite{AndoniKR18}, who established that any finite-dimensional normed space $\mathcal{X}=(\mathbb{R}^n, \|\cdot\|_{\mathcal{X}})$ admitting an efficient classical linear sketch must also admit a low-distortion linear embedding into $L_p$ for $p \in\left(\frac{2}{3},1\right)$ (and into $L_1$ if the norm is closed under sum-product). 
Because this embedding is a linear map, we can immediately integrate it within our turnstile streaming framework. 
Specifically, given a static embedding matrix $\bA \in \mathbb{R}^{M \times n}$, each stream update to the underlying vector $\bx$ implicitly updates the coordinates of the corresponding vector $\bA\bx$ in the embedded space. 
We can then execute our adversarially robust $L_p$ estimation algorithm directly on this virtual stream.

The major issue is that the continuous embedding from \cite{AndoniKR18} is existential and not necessarily computable. 
To make this algorithmic, we operate in an advice model and prove that the continuous embedding can be discretized. 
Fortunately, by using standard techniques, it can be shown the target dimension of the virtual stream can be upper bounded by $M = \O{n\log n}$, and that the entries of $\bA$ require only $\O{\log n}$ bits of precision. 
Since the space complexity of our robust $L_p$ approximation algorithm depends only polylogarithmically on $M$, we can run it on the transformed stream to establish a weak equivalence between classical sketching and adversarial robustness.

\paragraph{Derandomizing embeddings for EMD and clustering.}
This approach directly gives robust algorithms for important tasks such as approximating the Earth Mover Distance (EMD) over the grid $[\Delta]^d$. 
In the classical oblivious setting, EMD is estimated by embedding the spatial distributions into $L_1$ via randomly shifted quadtrees, achieving $\O{d \log \Delta}$ distortion~\cite{Indyk04,BackursIRW16,Cohen-AddadWZ23}. 
However, this randomized embedding $\bA_s$ only provides a ``for each'' guarantee, so that an adaptive adversary could observe the algorithm's outputs, deduce the random shift $\bs$, and generate updates that force worst-case distortion.

To prevent adversarial attacks against the embedding, we instead completely derandomize the construction. 
Namely, we define a deterministic embedding $\bA\bx = \frac{1}{n} \bigoplus_{s \in [\Delta]^d} \bA_s\bx$ that concatenates the scaled linear embeddings induced by \emph{all} possible $\Delta^d$ grid shifts simultaneously. 
By linearity of expectation, the $L_1$ norm of this concatenated vector exactly matches the expected value over the random shifts, providing a universal ``for all'' guarantee that holds against any adaptive adversary. 
While aggregating over all shifts $\bs\in[\Delta]^d$ blows up the embedding dimension from $\O{n}$ to $M =\O{n^2}$, where $n = \Delta^d$, our robust $F_p$ estimation algorithm again only uses space polylogarithmic in $M$. 

We naturally extend this robust EMD primitive to the $k$-median clustering problem, which can be formulated as the Earth Mover Distance between the data distribution $\mu$ and an assignment vector $\nu_{\calC}$ that routes mass to the nearest centers in a candidate set $\calC$. 
Using our deterministic $L_1$ embedding, we can dynamically maintain a robust sketch of the dataset. 
Because the embedding is deterministic and universally preserves distances, we can safely enumerate over a standard collection of candidate center sets and evaluate their robust costs on the fly, achieving the first robust algorithm for $k$-median clustering on turnstile streams.

\section{Preliminaries}

\subsection{Notation}
Throughout, we write $[n]$ to refer to the set of integers $\{1,2,\ldots,n\}$ for any positive integer $n$. 
The notation $\poly(n)$ refers to some polynomial in $n$ with fixed, but unspecified, degree. 
An event is said to happen with high probability if its likelihood is at least $1 - \frac{1}{\poly(n)}$. 
We typically use normal font for scalars and bold font for vectors or matrices.  

We typically use normal font for scalars and bold font for vectors or matrices. 
For vectors $\bx, \by \in \mathbb{R}^n$, we write $\langle \bx, \by \rangle$ for the standard Euclidean inner product. 
For $p > 0$, the $L_p$ norm (or quasi-norm when $p < 1$) of $\bx \in \mathbb{R}^n$ is defined as
\[\|\bx\|_p = \left(\sum_{i=1}^n |x_i|^p\right)^{1/p}.\]
For a vector $\bx$ and a subspace $V$, $\Proj_V(\bx)$ denotes the orthogonal projection of $\bx$ onto $V$, i.e., the vector in $V$ closest to $\bx$ in Euclidean distance. 
For a set of vectors $S$, $\Span(S)$ denotes the linear span of $S$, i.e., all linear combinations of vectors in $S$.  

We use $\bigoplus$ to denote vector concatenation (direct sum). 
For vectors $\bx^{(1)} \in \mathbb{R}^{n_1}, \ldots, \bx^{(k)} \in \mathbb{R}^{n_k}$, we define
\[
\bigoplus_{i=1}^k \bx^{(i)} = (\bx^{(1)}, \ldots, \bx^{(k)}) \in \mathbb{R}^{n_1 + \cdots + n_k}.
\]

We use $\circ$ to denote column-wise concatenation. 
For vectors $\ba_1, \ldots, \ba_k \in \mathbb{R}^n$, the matrix $\ba_1 \circ \cdots \circ \ba_k \in \mathbb{R}^{n \times k}$ is defined by taking $\ba_j$ as its $j$-th column. 
For a matrix $\bA \in \mathbb{R}^{n \times k}$, we write $\bA^\dagger$ to denote its Moore-Penrose pseudoinverse, defined as the unique matrix satisfying
\[\bA \bA^\dagger \bA = \bA, \quad \bA^\dagger \bA \bA^\dagger = \bA^\dagger, \quad (\bA \bA^\dagger)^\top = \bA \bA^\dagger, \quad (\bA^\dagger \bA)^\top = \bA^\dagger \bA.\]
When $\bA$ has full column rank, this reduces to $\bA^\dagger = (\bA^\top \bA)^{-1} \bA^\top$.

\subsection{Preliminaries for Small Moment Estimation}
We recall the existence of an embedding of $F_p$ into $F_2$.
\noindent

\begin{restatable}{theorem}{thmembed}
\label{thm:embed}
\cite{schoenberg1938metric,schoenberg1935definition,robertson2024negative}
For each $p\in(0,2)$, there exists an isometric embedding $f$ onto a Hilbert space $H$, such that
\[\|\bu-\bv\|_p^p=\|f(\bu)-f(\bv)\|_2^2,\qquad \|f(\mathbf{0})\|_2=0.\]
\end{restatable}
Similarly, we describe an embedding for $F_0$ into $F_2$ as follows. 
Given a vector $\bx\in\mathbb{R}^n$, suppose all coordinates of $x$ are bounded within $[-m,-(m-1),\ldots,-1,0,1,\ldots,m-1,m]$. 
Now for each coordinate $i \in [n]$ and each possible value $a$ that coordinate $i$ may take, introduce a distinct orthonormal basis vector $e_{i,a}$ in a high-dimensional Euclidean space. 
Then, we define
\[g(\bx) = \frac{1}{\sqrt{2}}\sum_{i=1}^n e_{i, x_i},\]
so that each coordinate of $\bx$ is represented by a one-hot vector indicating its value. 
Now for any two vectors $\bx,\by$ whose coordinates are within this range, we then have 
\begin{align*}
\|g(\bx) - g(\by) \|_2^2 &= \frac{1}{2}\sum_{i=1}^n \| e_{i, x_i} - e_{i, y_i} \|_2^2\\
&= \frac{1}{2}\cdot 2 \cdot |\{ i : x_i \neq y_i \}|\\
&= \|\bx - \by \|_0.
\end{align*}

However, note that for $p=0$, we deterministically have $\|g(\bx)\|_2^2=\frac{n}{2}$ for all $\bx\in\mathbb{R}^n$. To account for this, we further define the embedding by

\[f(\bx) \coloneq g(\bx) - g(\mathbf{0})\]

In particular, observe that this embedding preserves the $L_0$ norm as before, and maps $\mathbf{0}$ to $\mathbf{0}$: \[\|f(\bx) - f(\by)\|_2^2 = \|\bx - \by\|_0 \textrm{ and } \|f(\mathbf{0})\|_2^2 = 0\]

Thus, the one-hot encoding yields an isometric embedding of $(\mathbb{R}^n, \|\cdot\|_0)$ into $(\mathbb{R}^D, \|\cdot\|_2^2)$, where each differing coordinate contributes a fixed amount to the squared Euclidean distance. 
Hence, we have:
\begin{restatable}{theorem}{thmembedlzero}
\label{thm:embed:lzero}
There exists an isometric embedding $f$ from $\mathbb{R}^n$ onto a Hilbert space $H$, such that
\[\|\bu-\bv\|_0=\|f(\bu)-f(\bv)\|_2^2,\qquad \|f(\mathbf{0})\|_2=0.\]
\end{restatable}

Next, we recall oblivious turnstile streaming algorithms for $F_p$ moment estimation and distinct element estimation:
\begin{theorem}
\label{thm:obliv:small:fp}
\cite{Indyk06,Li08,KaneNW10a,KaneNW10b,KaneNPW11}
Given $p\in[0,2]$ and $\eps\in(0,1)$, there exists a linear sketching turnstile streaming algorithm $\FPSmall$ that uses $\O{\frac{1}{\eps^2}\log^2 n \log \frac{1}{\delta}}$ bits of space on a universe of size $n$ and a stream of length $m=\poly(n)$ and with probability $1-\delta$, outputs a $(1+\eps)$-approximation to the $F_p$ moment. 
\end{theorem}

\begin{figure*}[!htb]
\begin{mdframed}
\textbf{Algorithm $\DotEst$ to estimate $\langle f(\bx),f(\bq)\rangle$}:
\begin{enumerate}
\item 
Input: Vectors $\bx,\bq\in\mathbb{R}^n$
\item
Let $\widetilde{X}$ be a $(1+\eps)$-approximation of $\|\bx\|_p^p$ using $\FPSmall$
\item
Let $\widetilde{Q}$ be a $(1+\eps)$-approximation of $\|\bq\|_p^p$ using $\FPSmall$
\item
Let $\widetilde{D}$ be a $(1+\eps)$-approximation of $\|\bx-\bq\|_p^p$ using $\FPSmall$
\item
Output: $\frac{1}{2}(\widetilde{X}+\widetilde{Q}-\widetilde{D})$ as an estimate for $\langle f(\bx),f(\bq)\rangle$
\end{enumerate}
\end{mdframed}
\caption{Estimator for dot product in the embedded space.}
\label{fig:dotest}
\end{figure*}
We now show the correctness of the algorithm $\DotEst$, which estimates the dot product of two vectors in the embedded space. 
\begin{lemma}
\label{lem:dotest}
With high probability, $\DotEst$ outputs an estimate to $\langle f(\bx),f(\bq)\rangle$ with additive error $\eps\cdot\|\bx\|_p^p+\eps\cdot\|\bq\|_p^p+\eps\cdot\|\bx-\bq\|_p^p$. 
The algorithm uses $\O{\frac{1}{\eps^2}\log^3 n}$ bits of space. 
\end{lemma}
\begin{proof}
By the correctness of $\FPSmall$, we have:
\begin{align*}
\left\lvert\widetilde{X}-\|\bx\|_p^p\right\rvert&\le\eps\cdot\|\bx\|_p^p\\
\left\lvert\widetilde{Q}-\|\bq\|_p^p\right\rvert&\le\eps\cdot\|\bq\|_p^p\\
\left\lvert\widetilde{D}-\|\bx-\bq\|_p^p\right\rvert&\le\eps\cdot\|\bx-\bq\|_p^p.
\end{align*}
By triangle inequality, it follows that
\[\left\lvert\widetilde{X}+\widetilde{Q}-\widetilde{D}-\|\bx\|_p^p-\|\bq\|_p^p+\|\bx-\bq\|_p^p\right\rvert\le\eps\cdot\|\bx\|_p^p+\eps\cdot\|\bq\|_p^p+\eps\cdot\|\bx-\bq\|_p^p.\]
By \Cref{thm:embed}, we have 
\[\|\bx-\bq\|_p^p=\|f(\bx)-f(\bq)\|_2^2=\|f(\bx)\|_2^2+\|f(\bq)\|_2^2-2\cdot\langle f(\bx),f(\bq)\rangle.\]
Hence by \Cref{thm:embed},
\begin{align*}
\langle f(\bx),f(\bq)\rangle&=\frac{1}{2}\left(\|f(\bx)\|_2^2+\|f(\bq)\|_2^2-\|f(\bx)-f(\bq)\|_2^2\right)\\
&=\frac{1}{2}\left(\|\bx\|_p^p+\|\bq\|_p^p-\|\bx-\bq\|_p^p\right).
\end{align*}
Therefore,
\[\left\lvert\frac{1}{2}(\widetilde{X}+\widetilde{Q}-\widetilde{D})-\langle f(\bx),f(\bq)\rangle\right\rvert\le\eps\cdot\|\bx\|_p^p+\eps\cdot\|\bq\|_p^p+\eps\cdot\|\bx-\bq\|_p^p.\]
The space complexity follows from requiring $(1+\eps)$-approximation to the $F_p$ moment, using the subroutine from \Cref{thm:obliv:small:fp}. 
\end{proof}

\begin{lemma}
\label{lem:coefficient:sum}
Let $\bA$ be the $n \times L$ matrix where the $i$-th column is $\bA_i = \bq_i$. 
Let $\by = \arg \min_{\bx \in \mathbb{R}^L} \|\bA \bx - \bz \|_2^2$ subject to the constraint $\sum_{i = 1}^L x_i = 1$. 
Then, $\|\bA \by \|_2 \leq \|\bz - \bq_1\|_2 + \|\bz\|_2$. 
\end{lemma}
\begin{proof}
Define matrix $\bB \in \mathbb{R}^{n \times L}$ whose $i$-th column is given by $\bB_i = \bq_i - \bq_1$. 
Note that for any $\bx$ such that $\sum_{i =1}^L x_i = 1$, we have
\[\bB \bx = \sum_{i = 1}^L (\bq_i - \bq_1) \bx_i = \sum_{i = 1}^L x_i \bq_i - \left(\sum_{i = 1}^L x_i \right) \bq_1 = \bA \bx - \bq_1. \]
Consequently, it follows that $\by = \arg \min_{\bx \in \mathbb{R}^L} \|\bB \bx - (\bz - \bq_1)\|_2^2$. 
Since $\by$ is optimal, 
\[\|\bB \by - (\bz - \bq_1)\|_2 \leq \|\bz - \bq_1\|_2.\]
Thus, we have that 
\[\|\bA\by\|_2 = \|\bB \by + \bq_1\|_2 \leq \|\bB \by - (\bz - \bq_1)\|_2 + \|\bz\|_2 \leq \|\bz - \bq_1\|_2 + \| \bz\|_2,\]
as desired.
\end{proof}

\section{Robust \texorpdfstring{$F_p$}{Fp} Algorithm for \texorpdfstring{$p\in[0,2)$}{p in [0,2)}}
In this section, we present our adversarially robust algorithm for $F_p$ moment estimation on insertion-deletion streams, for $p\in[0,2)$. 
Our goal is to generalize the estimator, corrector, and learner framework developed for $F_2$ moment estimation in \cite{GribelyukLWYZ26}. 
The key challenge is that $L_p$ spaces for $p \neq 2$ are no longer Hilbert spaces, i.e., there is no inner product that induces the $L_p$ norm, so geometric reasoning based on orthogonality or the Pythagorean theorem does not apply. Instead, our analysis relies on the existence of an embedding of $L_p^p$ into $L_2^2$.

\subsection{Subroutines for Small Moment Estimation}
We next describe an algorithm for $L_p$ norm estimation, with $p\in[0,2)$. 
Utilizing \Cref{thm:embed}, there exists an embedding $f$ from $F_p$ to $F_2$, so that we have
\[\|\bx-\by\|_p^p=\|f(\bx)-f(\by)\|_2^2=\|f(\bx)\|_2^2-2\langle f(\bx),f(\by)\rangle+\|f(\by)\|_2^2.\]

Note that for all $p \in [0,2)$ we have $f(\mathbf{0})=\mathbf{0}$, so that 
\[\|f(\bx)\|_2^2=\|f(\bx)-f(\mathbf{0})\|_2^2=\|\bx-\mathbf{0}\|_p^p=\|\bx\|_p^p\]
and thus
\[\|\bx-\by\|_p^p=\|\bx\|_p^p-2\langle f(\bx),f(\by)\rangle+\|\by\|_p^p.\]
Then it follows that we can use $F_p$ estimation algorithms to estimate $\|\bx\|_p^p$ and $\|\by\|_p^p$. 
Moreover, since $\|\bx-\by\|_p^p=\|f(\bx)-f(\by)\|_2^2$, then we can further estimate $2\langle f(\bx),f(\by)\rangle$ using $F_p$ estimation algorithms, such as in $\DotEst$. 

In this subsection, we design and analyze a robust algorithm which estimates the $F_p$ moment up to $\eps \cdot A$ additive error when the underlying input vector satisfies $\|\bx\|_p^p = \O{A}$. Then, in \Cref{sec:mult:error} we show how to modify our $F_p$ estimation algorithm to obtain the $(1+\eps)$ multiplicative error guarantee.

\paragraph{Estimator, corrector, and learner.} 
As in \cite{GribelyukLWYZ26}, we have a learner that maintains an approximation $\bv$ to $f(\bz)$. 
By analogy, we can update $\bv$ by adding $\eps \cdot (\bv + f(\bq))$ each time that the estimator is incorrect on a query $\bq$; however, in this setting, this update rule may result in large coefficients in the linear combination $\bv = \sum_{i \in [L]} \alpha_i \cdot f(\bq_i)$. 
Thus, we modify the previous update rule as follows: we add the vector $\bq$ to a list $\calL_i$ and solve a kernel regression problem to find the vector $\bv$ in the span of $\calL_i$ that best approximates $\bz$. 
Crucially, we embed all vectors into $L_2$ space so that it suffices to solve an $L_2$ regression problem in the kernel space, where we can compute the closed-form solution using a linear combination of dot products of either the form $\langle f(\bz),f(\bq_i)\rangle$ or $\langle f(\bq_i),f(\bq_j)\rangle$, which can be computed using $\O{\frac{1}{\eps^2} \log^3 n}$ bits of space by \Cref{lem:dotest}.

\begin{algorithm}[!htb]
\caption{$\EstLevel(i)$ for $F_p$ moment estimation}
\label{alg:est:level:smallp}
\begin{algorithmic}[1]
\State{Let $\calP_i$ be the active block at level $i$}
\State{Let $\bz$ be the vector of the stream before $\calP_i$}
\State{Let $\bq$ be the active query in $\calP_i$}
\State{Subroutine $\FPSmall$ outputs a $\left(1+\frac{\eps^2}{100H^2}\right)$-approximation for $F_p$ moment for $\|\bz -\bq\|_p^p = \Theta(A)$, robust to $\poly\left(\frac{1}{\eps},\log n\right)$ adaptive queries}
\If{$i=1$}
\State{$P_i\gets 0$}
\State{$Q_i\gets\EmbedEst(\bq,\bz)$}
\Else
\State{Let $\calL_i$ be the current list from $\MaintainList$ for the active block at level $i$}
\State{Let $\bv$ be the current iterate from $\GetIterate(\bz,\calL_i)$}
\State{$P_i\gets\EmbedEst(\bz,\bv) + \nu^2 \|\alpha\|_2^2$} \Comment{$\alpha$ is the vector of coefficients of $f(\bq_i)$ in $\bv = \sum_{ i = 1}^L \alpha_i f(\bq_i)$.}
\State{$Q_i\gets\EstLevel(i-1) +\nu^2 \|\alpha\|_2^2$}
\Comment{Estimate $\|\bv - f(\bq)\|_2^2$ as in \Cref{lem:recursion:computation}, recurse on $\|\sum_{i = 1}^L \alpha_i \cdot f(\bq_i - \bq_L) - f(\bq_{L+1})\|_2^2$, which is passed to $\EstLevel(i-1)$.}
\State{Let $A_i$ be the output of $\FPSmall(\bz-\bq)$}
\If{$(P_i+Q_i) \in A_i \pm \frac{6i\eps}{100 H} A$} 
\State{\Return $P_i+Q_i$.}
\Else
\State{\Return $A_i$}
\EndIf
\EndIf
\end{algorithmic}
\end{algorithm}

\begin{algorithm}[!htb]
\caption{$\MaintainList(i)$ for $F_p$ moment estimation}
\label{alg:maintain:list:smallp}
\begin{algorithmic}[1]
\State{Let $\calP_i$ be the active block at level $i$}
\State{Let $\bz$ be the stream corresponding to all updates prior to $\calP_i$}
\State{Let $\bq$ be the part of the query in block $\calP_i$}
\State{$Z_i\gets\EstLevel(i)$}
\State{Let $\calA_i$ be an $F_p$ moment estimation algorithm}\Comment{Robust to $\poly\left(\frac{1}{\eps},\log n\right)$ adaptive queries}
\If{$Z_i\notin\left[\calA_i(\bz-\bq)- \frac{6(i+1)\eps A}{100H},\calA_i(\bz-\bq)+ \frac{6(i+1)\eps A}{100H}\right]$}
\State{Add $\bq$ to $\calL_{i+1}$}
\EndIf

\end{algorithmic}
\end{algorithm}

Formally, the procedure $\EmbedEst$ in \Cref{alg:embed:est} computes the squared distance $\|f(\bz)-\bv\|_2^2$ in the kernel space using estimates of the norms and pairwise dot products. 
The procedure $\EstLevel$ in \Cref{alg:est:level:smallp} then recursively estimates the $F_p$ moment at each recursive level $i$ by combining the embedded approximation from $\EmbedEst$ with the estimate from the previous level. 
Meanwhile, $\MaintainList$ in \Cref{alg:maintain:list:smallp} maintains a list $\calL_i$ at level $i$, which includes each query vector $\bq$ for which the corrector flags the output from the estimator as being inaccurate. 
When this occurs, the subroutine $\GetIterate$ in \Cref{alg:get:iterate} computes $\bv$ by finding a near-optimal approximation of $f(\bz)$ in the span of the list $\calL_i$. 

\paragraph{Recursion.}
In the top-level of the algorithm, we use the \DotEst subroutine to estimate $\|f(\bz) - \bv\|_2^2$. Following the robust $L_2$ algorithm of \cite{GribelyukLWYZ26}, we need to recursively compute the remaining quantity $\|\bv - f(\bq)\|_2^2$ in the next level, where the evolving portion of the stream $\bq = \bq_L + \bq_{L+1}$ is divided into blocks which must be sketched independently. Unfortunately, we do not have access to the sketch of $\bv$ in the embedded space, and in fact, $\bv$ may not have a pre-image in the original space. This significantly complicates our implementation of the recursion and related computations. In the following lemma, we show that all terms of $\|\bv - f(\bq)\|_2^2$ can be computed robustly, and that it suffices to recursively compute only \textit{one} squared-norm. 

\begin{lemma}\label{lem:recursion:computation}
    At each level $i$ of the recursion in \Cref{alg:est:level:smallp}, $Q_i = \EstLevel(i-1)$ executes a single recursive call to estimate the squared norm $\left \|\sum_{i = 1}^L \alpha_i f(\bq_i - \bq_{L}) - f(\bq_{L + 1})\right \|_2^2$, where $\calL_i = \{f(\bq_i)\}_{i \in [L]}$ are the previous queries on which the estimator was incorrect and the current iterate is $\bv = \sum_{i = 1}^L \alpha_i\cdot f(\bq_i)$.
\end{lemma}
\begin{proof}
    
We expand $\|\bv - f(\bq)\|_2^2$ and describe how each term can be computed robustly.
\begin{align*}
\|\bv - f(\bq)\|_2^2&= \left\|\sum_{i \in [L]} \alpha_i \cdot f(\bq_i) - f(\bq) \right\|_2^2  \\ &
= \sum_{1 \leq i,j \leq L} \left(\alpha_i \cdot \alpha_j \cdot \langle f(\bq_i), f(\bq_j) \rangle \right) - 2 \sum_{i \in [L]} \alpha_i \langle f(\bq_i), f(\bq) \rangle + \|f(\bq)\|_2^2
\end{align*}
We consider each of the three types of terms above. We begin with the simplest part: note that since past queries $\bq_i \in \calL$ are fixed, the inner products of the form $\langle f(\bq_i), f(\bq_j) \rangle$ can be computed robustly by invoking the $\DotEst$ subroutine directly. 

Next, we can write the last two terms as follows:
\begin{align*}
    -2\sum_{i \in [L]} \alpha_i \langle f(\bq_i), f(\bq) \rangle + \|f(\bq)\|_2^2 &=  - \sum_{i = 1}^L \alpha_i \left(\|f(\bq_i)\|_2^2 + \|f(\bq)\|_2^2 - \|f(\bq_i) - f(\bq)\|_2^2 \right) + \|f(\bq)\|_2^2 \\
    & = \sum_{i= 1}^L \alpha_i \|f(\bq_i) - f(\bq)\|_2^2 - \sum_{i = 1}^L \alpha_i \|f(\bq_i)\|_2^2 + \left(1-\sum_{i = 1}^L \alpha_i\right) \cdot \|f(\bq)\|_2^2
\end{align*}

The term $\sum_{i = 1}^L \alpha_i \cdot \|f(\bq_i)\|_2^2$ consists only of previous fixed queries $\bq_i \in \calL$ and thus can be computed directly by using \FPSmall sketches. 
Additionally, since the coefficients of $\bz' = \sum_{i = 1}^L \alpha_i f(\bq_i)$ satisfy $\sum_{i = 1}^L \alpha_i = 1$, the third term above is identically zero. 
Therefore, it suffices to robustly compute the value of $\sum_{i = 1}^L \alpha_i \|f(\bq_i) - f(\bq)\|_2^2$. 
To this end, observe that since $f$ is an isometric embedding into $L_2^2$, we have that \[\|f(\bq_i) - f(\bq) \|_2^2 = \|\bq_i - \bq_L - \bq_{L+1}\|_p^p = \|f(\bq_i - \bq_L) -f(\bq_{L+1})\|_2^2.\]

Thus, the remaining term can be expressed as
\begin{align*}
\sum_{i = 1}^L \alpha_i \cdot \|f(\bq_i) - f(\bq)\|_2^2 &= \sum_{i = 1}^L \alpha_i \cdot  \|f(\bq_i - \bq_L) - f(\bq_{L+1})\|_2^2 \\ 
&= \sum_{i =1}^L \alpha_i \cdot \|f(\bq_i - \bq_L)\|_2^2 - 2 \left \langle \sum_{i = 1}^L \alpha_i f(\bq_i - \bq_L), f(\bq_{L+1}) \right \rangle  + \|f(\bq_{L+1})\|_2^2.
\end{align*}
At this point, we recall that our $L_2$ algorithm initializes a new block each time that the estimator at that level is inaccurate. 
By construction, we begin a new block every time that a query is added to the list, so in particular $\bq_L \in \calL$ and $\{\bq_i - \bq_L\}_{i \in [L]}$ and $\bq_{L+1}$ can be sketched using independent sketches. 
Moreover, each $\bq_i - \bq_{L}$ is a fixed vector so we can simply use \FPSmall to estimate each of the norms $\|f(\bq_i - \bq_L)\|_2^2$. 
We consider the last two terms:
\begin{align*}
&- 2 \left \langle  \sum_{i = 1}^L \alpha_i f(\bq_i - \bq_L), f(\bq_{L+1}) \right \rangle  + \|f(\bq_{L+1})\|_2^2 \\
& = \left \|\sum_{i = 1}^L \alpha_i f(\bq_i - \bq_L) - f(\bq_{L+1}) \right \|_2^2 - \|f(\bq_{L+1})\|_2^2 - \left \|\sum_{i = 1}^L \alpha_i f(\bq_i - \bq_L) \right \|_2^2  + \|f(\bq_{L+1})\|_2^2 \\
& = \left \|\sum_{i = 1}^L \alpha_i f(\bq_i - \bq_L) - f(\bq_{L + 1}) \right \|_2^2 - \left \|\sum_{i = 1}^L \alpha_i f(\bq_i - \bq_L) \right \|_2^2
\end{align*}
As before, the second term only depends on the fixed vectors $\{\bq_i - \bq_L\}_{i \in [L]}$ so it can be estimated by invoking the \DotEst subroutine. 
Finally, we can recursively estimate the first term by applying our robust $L_p$ algorithm. 
\end{proof}

\begin{remark}
The input to the recursion at each level will be of the form 
\[\left\|\sum_{i = 1}^L \alpha_i \cdot f(\bq_i - \bq_L) - f(\bq_{L+1})\right\|_2^2,\]
and in particular the estimator which is computed in the next level is given by 
\[\left\|\sum_{ i = 1}^L \alpha_i \cdot f(\bq_i - \bq_L) - \bv'\right\|_2^2 + \|\bv' - f(\bq_{L+1})\|_2^2,\]
where $\bv'$ is the iterate at the next level. 
Importantly, note that $\bv' = \sum_{i = 1}^L \beta_i \cdot f(\bq_i')$s represents the projection of $\sum_{i = 1}^L \alpha_i \cdot f(\bq_i - \bq_L)$ onto the list of previous incorrect queries $\{\bq_i'\}_{i \in [L']}$ of the next level, and can be estimated by computing $\O{L^2}$ pair-wise inner products between list vectors in the adjacent levels.
\end{remark}

In order to implement the corrector at each level of the recursion, in \Cref{lem:recursion:bounded:norm} we show that the input to the recursion at each level has bounded norm.

\begin{algorithm}[H]
\caption{$\EmbedEst(\bz,\bv)$ for $F_p$ moment estimation}
\label{alg:embed:est}
\begin{algorithmic}[1]
\State{Let $\bv=\alpha_1 f(\bq_1)+\ldots+\alpha_L f(\bq_L)$}
\State{Let $Z$ be an estimate of $\|\bz\|_p^p$ from $\FPSmall$}
\State{Let $U_1,\ldots,U_L$ be estimates of $\|\bq_1\|_p^p,\ldots,\|\bq_L\|_p^p$ from $\FPSmall$}
\State{Let $\{C_i\}$ be estimates of $\langle f(\bz),f(\bq_i)\rangle$ using $\EstDot$}
\State{Let $\{D_{i,j}\}$ be estimates of $\langle f(\bq_i),f(\bq_j)\rangle$, with $i\neq j$, using $\EstDot$}
\State{Use these quantities to return an estimate of $\|f(\bz)-\bv\|_2^2$}
\end{algorithmic}
\end{algorithm}

\begin{algorithm}[H]
\caption{$\GetIterate(\bz,\calL_i)$ for $F_p$ moment estimation}
\label{alg:get:iterate}
\begin{algorithmic}[1]
\State{Let $L$ be an upper bound on $|\calL_i|$, $\nu_0\gets\O{\frac{\eps^2}{H^2 \cdot \sqrt{L}}}$, and set $\nu \gets \nu_0 \sqrt{A}$.}
\State{Let $\bu_j=f(\bq_j)+\nu\cdot\be_{N+j} \in \mathbb{R}^{N + L}$}
\State{Let $\calL_i' = \{\bu_j-\bu_1 \ | \ f(\bq_j) \in \calL_i \}$}
\State{Let $Z$ be an estimate of $\|\bz\|_p^p$ from $\FPSmall$}
\State{Let $U_1,\ldots,U_L$ be estimates of $\|\bq_1\|_p^p,\ldots,\|\bq_L\|_p^p$ from $\FPSmall$}
\State{Let $\{C_j\}$ be estimates of $\langle f(\bz),f(\bq_j)\rangle$ using $\EstDot$}
\State{Let $\{D_{j,k}\}$ be estimates of $\langle f(\bq_j),f(\bq_k)\rangle$, with $j\neq k$, using $\EstDot$}
\State{Use these quantities to return $\argmin_{\bu'\in\Span(\calL_i')}\|f(\bz)-\bu_1-\bu'\|_2^2$}
\State{Let $\bu'=\sum_{j=1}^L\beta_j\bu_j$}
\State{Return coefficients of $\bu_1 + \bu'$ as coefficients of $f(\bq_1),\ldots, f(\bq_{L})$.}

\end{algorithmic}
\end{algorithm}
\FloatBarrier

\begin{lemma}[Bounding the norm]\label{lem:recursion:bounded:norm}
For any level $k \in [H]$, let $\bw_k$ be the input to the recursion, let $\bv$ be the iterate at level $k$ which approximates $f(\bz)$. Then $\|\bw_k\|_2^2 \leq \|\bz\|_p^p + \O{A}$.
\end{lemma}
\begin{proof}
At each level $k \in [H]$, the input to the recursion at the next level is 
\[\|\bw_k\|_2^2 = \left \|\sum_{i = 1}^{L} \alpha_i f(\bq_i - \bq_L) - f(\bq_{L+1}) \right \|_2^2,\]
where $\calL_i = \{f(\bq_i)\}_{i = 1}^L$ denotes the list at level $k$. 
Furthermore, recall that $\bv = \sum_{i = 1}^L \alpha_i \cdot f(\bq_i)$ is the projection of $f(\bz - \bq_1)$ onto the subspace spanned by the set of list vectors $f(\bq_i)$ for $i \in [L]$, and by \Cref{lem:coefficient:sum} we have $\sum_{i = 1}^L \alpha_i = 1$. Then, we have

\begin{align*}
\left \|\sum_{i = 1}^L \alpha_i f(\bq_i - \bq_L) \right \|_2^2  &= \sum_{i,j = 1}^{L} \alpha_i \cdot \alpha_j \langle f(\bq_i- \bq_L), f(\bq_j - \bq_L)\rangle \\ 
&= \frac{1}{2} \sum_{i,j = 1}^{L} \alpha_i\cdot \alpha_j \left(\|f(\bq_i) - f(\bq_L)\|_2^2 + \|f(\bq_j) - f(\bq_L)\|_2^2 - \|f(\bq_i) - f(\bq_j)\|_2^2  \right) \\ 
&= \frac{1}{2} \sum_{i, j = 1}^L \left(-2 \langle f(\bq_i), f(\bq_L) \rangle - 2\langle f(\bq_j), f(\bq_L) \rangle + 2 \langle f(\bq_i), f(\bq_j) \rangle + 2 \|f(\bq_L)\|_2^2  \right) \\ 
&= \sum_{i,j = 1}^L \alpha_i \cdot \alpha_j \langle f(\bq_i), f(\bq_j) \rangle - 2 \sum_{i = 1}^L \alpha_i \cdot \langle f(\bq_i), f(\bq_L) \rangle + \|f(\bq_L)\|_2^2 \\ 
&= \left\|\sum_{i = 1}^L \alpha_i f(\bq_i) - f(\bq_L)\right\|_2^2 = \|\bv - f(\bq_L)\|_2^2.
\end{align*}

In particular, it follows that 

\[\left \| \sum_{i = 1}^L \alpha_i \cdot f(\bq_i - \bq_L) \right \|_2 = \|\bv - f(\bq_L)\|_2 \leq \|\bv\|_2 + \|f(\bq_L)\|_2.\]
 Recall that by assumption, we have $\|f(\bq_L)\|_2 = \O{\sqrt{A}}$ and $\|f(\bq_{L+1})\|_2 = \O{\sqrt{A}}$. By triangle inequality, we conclude that 

\begin{align*}
    \left \|\sum_{i = 1}^L \alpha_i f(\bq_i - \bq_L)- f(\bq_{L+1}) \right \|_2 &\leq \|\bv - f(\bq_L)\|_2 + \|f(\bq_{L+1})\|_2 \\ & \leq \|\bv\|_2 + \|f(\bq_L)\|_2 + \|f(\bq_{L+1})\|_2  \\ & \leq \|\bv\|_2 + \O{\sqrt{A}}
\end{align*} 

Thus, $\left \|\sum_{i = 1}^L \alpha_i f(\bq_i - \bq_L) - f(\bq_{L+1}) \right \|_2^2 \leq \|\bv\|_2^2 + \O{A} + \|\bv\|_2 \cdot \O{\sqrt{A}}$ . By \Cref{lem:coefficient:sum} we know that $\|\bv\|_2^2 \leq \left(\|f(\bz)\|_2 + \|f(\bz) - f(\bq_1)\|_2\right)^2 \leq \|\bz\|_p^p + \O{A}$, which implies the claim.
\end{proof}

\paragraph{Update rule and regularized regression.}
Next, we discuss our modified update rule for the iterate $\bv$ in each level of the algorithm. 
In particular, let $\calL = \{f(\bq_i)\}_{i \in [L]}$ be the list of queries on which the algorithm was previously inaccurate. 
Then, we note that the $L_2$ algorithm of \cite{GribelyukLWYZ26} can be implemented by following the alternative update rule $\bv = \arg \min_{\by \in \mathbb{R}^L}{\|\bU \by - f(\bz)\|_2^2}$, where $\bU$ is the matrix with columns $f(\bq_i)$. Then, by expressing $\bv = \sum_{i = 1}^L \alpha_i \cdot f(\bq_i)$, we can approximate the estimator $\|f(\bz) - \bv\|_2^2 + \|\bv - f(\bq)\|_2^2$ by estimating the necessary inner products of the form $\langle f(\bz), f(\bq_i) \rangle$ or $\langle f(\bq_i), f(\bq_j) \rangle$ for $i, j \in [L]$ via the \DotEst subroutine. For instance, for the first term of the estimator, we have 

\[\|f(\bz) - \bv\|_2^2 = \|f(\bz)\|_2^2 - 2 \sum_{i = 1}^L \alpha_i \cdot \langle f(\bz), f(\bq_i) \rangle + \sum_{1 \leq i,j \leq L} \alpha_i \cdot \alpha_j \langle f(\bq_i), f(\bq_j) \rangle. \]

Similarly, by \Cref{lem:recursion:computation}, the second term $\|\bv- f(\bq)\|_2^2$ can be expanded into inner products involving the previous fixed vectors $f(\bq_i)$ and $f(\bz)$, and a single squared norm on which we recurse. By \Cref{lem:dotest}, each inner product $\langle f(\bz), f(\bq_i) \rangle$ can be estimated with additive error $\eps \cdot \left(\|\bz\|_p^p + \|\bq_i\|_p^p + \|\bz - \bq_i\|_p^p \right)$ by using $F_p$ estimation sketches, and similarly, $\langle f(\bq_i), f(\bq_j) \rangle$ can be estimated with additive error $\eps \cdot \left(\|\bq_i\|_p^p + \|\bq_j\|_p^p + \|\bq_i - \bq_j\|_p^p \right)$. By assumption, we have that $\|\bz\|_p^p = \Theta(A)$, $\|\bz - \bq_i\|_p^p = \Theta(A)$ and all queries $\bq_i$ also satisfy $\|\bq_i\|_p^p \leq \O{A}$, so each such inner product can be approximated up to $\eps \cdot A$ additive error. 
Therefore, in order to upper bound the total error incurred by our estimates via \DotEst, we must ensure that the coefficients $\alpha_i$ are bounded by $\poly\left(\frac{1}{\eps}, \log n \right)$. Since this is not true in general, we must modify our update rule, as we describe next.

To this end, we instead consider the vectors $f(\bq_i)$ with \textit{added} dimensions, i.e. let $\bu_i = f(\bq_i) + \nu \cdot \be_{N + i} \in \mathbb{R}^{N+L}$ for $i \in [L]$, and let $\calL'$ be the list containing $\bu_1, \ldots, \bu_L$. Then, we observe that projecting $f(\bz)$ onto $\calL'$ exactly corresponds to the \textit{regularized} regression problem $\argmin_{\bu' \in \Span(\calL)} \left(\|f(\bz) - \bu' \|_2^2 + \nu^2 \cdot \|\mathbf{\alpha}\|_2^2\right)$, where $\alpha \in \mathbb{R}^L$ denotes the vector of coefficients of the iterate. Let $\bu' = \sum_{i=1}^L \alpha_i \cdot \bu_i$ be the solution to this regularized regression problem; then, for a hard query $\bq$, we update the current iterate $\bv$ to $\bv' = \sum_{ i = 1}^L \alpha_i \cdot f(\bq_i)$.

Finally, recall that in \Cref{lem:recursion:computation} we required that the sum of the coefficients satisfies $\sum_{ i =1}^L \alpha_i = 1$. To ensure this, we simply impose an additional linear constraint and instead project $f(\bz) - \bu_1$ onto the affine subspace $\bu_1 + \Span \{ \bu_i - \bu_1 : i = 2, \ldots,L \}$. The following lemmas allow us to bound the condition number of the vectors $\bu_1,\ldots, \bu_{L}$.

\begin{lemma}\label{lem:condition:number}
Suppose $\|\bu_j\|_2=\O{A}$ for all $j\in[L]$ and let $\bU'$ be the matrix whose columns are the vectors $\bu_2-\bu_1,\ldots, \bu_L-\bu_1$. 
Then $\sigma_{\max}(\bU')=\O{\sqrt{L\cdot A}}$ and $\sigma_{\min}(\bU')\ge\nu_0 \sqrt{A}$. 
\end{lemma}
\begin{proof}
Observe that since $\|\bu_j\|_2^2=\O{A}$, then by triangle inequality, $\|\bu_j-\bu_1\|_2\le\|\bu_j\|_2+\|\bu_1\|_2=\O{\sqrt{A}}$ for all $j\neq 1$. 
  
\[\sigma_{\max}(\bU')\le\|\bU\|_F\le\sqrt{L}\cdot\sqrt{\Theta(A)+\nu_0^2 A}=\O{\sqrt{L\cdot A}}.\] 
On the other hand, since the $L\times L$ identity matrix scaled by $\nu$ is appended in the dimensions after $N$, then we have $\sigma_{\min}(\bU')\ge\nu_0 \sqrt{A}$. 
\end{proof}

\begin{lemma}\label{lem:coefficient:bound}
Let $\bA = \ba_1 \circ \ldots \circ \ba_k \in \mathbb{R}^{n \times k}$ have full column rank and let $\bv \in \mathbb{R}^n$. 
Let $\bx^\star = \arg\min_{\bx \in \Span(\bA)} \|\bv - \bx\|_2$ and $\bx^\star = \bA \by^\star$. 
If $\bv \notin \Span(\bA)$, then the coefficients of the optimal representation satisfy
\[\|\by^\star\|_2 \le \frac{\|\bv\|_2}{\sigma_{\min}(\bA)}.\]
\end{lemma}
\begin{proof}
Since $\bA$ has full column rank, the minimizer of $\|\bv - \bA \by\|_2$ is unique and given by
\[\by^\star = (\bA^\top \bA)^{-1} \bA^\top \bv = \bA^\dagger \bv,\]
where $\bA^\dagger$ denotes the Moore-Penrose pseudoinverse of $\bA$. 
The corresponding vector $\bx^\star = \bA \by^\star$ is the orthogonal projection of $\bv$ onto $\Span(\bA)$. 
Using the operator norm of the pseudoinverse and sub-multiplicativity of norms, 
\[\|\by^\star\|_2 = \|\bA^\dagger \bv\|_2 \le \|\bA^\dagger\|_2 \cdot \|\bv\|_2.\]
Since $\|\bA^\dagger\|_2 = \frac{1}{\sigma_{\min}(\bA)}$, i.e., the minimum singular value of $\bA$, then it follows that
\[\|\by^\star\|_2 \le \frac{\|\bv\|_2}{\sigma_{\min}(\bA)}.\]
\end{proof}

\paragraph{Algorithm analysis.} 
Next, we justify the correctness guarantees of our algorithm, assuming that the underlying frequency vector $\bx$ satisfies $\|\bx\|_p^p \leq 10 A$ when the sketch is invoked. In the next lemma, we show that $P_i$ in \EstLevel is an accurate approximation for $\|f(\bz) - \bv\|_2^2 + \nu^2 \|\alpha\|_2^2$.

\begin{lemma}\label{lem:est:Pi}
Let $\bv$ be a fixed iterate vector in block $C_{i,j}$ at level $i$, i.e., the value of $\bv$ in \Cref{alg:est:level:smallp}, at a fixed time in the stream, conditioned on the previous times, and suppose that $\|f(\bz)\|_2^2 = \Theta(A)$ and $\|f(\bz) - f(\bq)\|_2^2 = \Theta(A)$. Suppose that $P_i$ is  produced using sketch $\bB_i$ which is robust for $\O{\frac{1}{\eta^2} \log n}$ adaptive updates, where $\eta = \Theta(\frac{\eps}{H})$.
Then with high probability, 
\[\|f(\bz)-\bv\|_2^2 + \nu^2 \|\alpha\|_2^2 -\frac{\eps^2 A}{100H^2}\le P_i\le \|f(\bz)-\bv\|_2^2 + \frac{\eps^2 A}{100 H^2} + \nu^2 \|\alpha\|_2^2.\]
\end{lemma}
\begin{proof}
Note that for $\bv = \sum_{i = 1}^L \alpha_i \cdot f(\bq_i)$, we have that 
\begin{align*}
   \|f(\bz) - \bv\|_2^2 &= \|f(\bz)\|_2^2 + \left \| \sum_{i = 1}^L \alpha_i f(\bq_i) \right \|_2^2 - 2 \sum_{i = 1}^L \alpha_i \cdot \langle f(\bz), f(\bq_i) \rangle  
\end{align*}
By \Cref{lem:condition:number} we have that the condition number of the list is $\kappa = \frac{\sqrt{L}}{\nu_0} = \frac{L\cdot H^2}{\eps^2}$, so it suffices for us to obtain additive error $\O{\frac{\eps^2}{\kappa^2 \cdot H^2 \cdot L^2}} \cdot A$ for each of the $\O{L^2}$ terms.

Observe that $\|f(\bz)\|_2^2 = \|\bz\|_p^p$ can be estimated using $\bB_i$ up to $\left(1+ \frac{\eps^2}{100H^2 L^2 \kappa^2} \right)$ multiplicative error. Thus, it suffices to estimate the remaining terms: each of the $L$ terms $\|f(\bq_i)\|_2^2 = \|\bq_i\|_p^p$ can be similarly estimated up to $\left(1+ \frac{\eps^2}{100H^2 \cdot L^2 \cdot \kappa^2} \right)$ multiplicative error using $\bB_i$, and for $i \not= j$ the remaining inner products $\langle f(\bq_i), f(\bq_j)\rangle $ are estimated up to $\frac{\eps^2}{100H^2 \cdot L^2 \cdot \kappa^2} \cdot \left( \|\bq_i\|_p^p + \|\bq_j\|_p^p + \|\bq_i - \bq_j\|_p^p \right) $ additive error by invoking the $\DotEst$ subroutine. Likewise, each of the $L$ terms $\langle f(\bz), f(\bq_i) \rangle$ is estimated up to $\frac{\eps^2}{100H^2 \cdot L^2 \cdot \kappa^2} \cdot \left(\|\bz\|_p^p + \|\bq_i\|_p^p + \|\bz - \bq_i\|_p^p \right)$ additive error via $\DotEst$. Since each of the coefficients $\alpha_i$ in the representation $\bv = \sum_{i = 1}^L \alpha_i \cdot f(\bq_i)$ is bounded by $\kappa$, the desired error bound follows by adding up all of these estimates and recalling that $\|\bq_i\|_p^p = \O{A}$ for all $i \in [L]$ and $\|\bz\|_p^p = \Theta(A)$. 
\end{proof}

Building on the previous lemma, we now argue that the output $P_i + Q_i$ of \EstLevel$(i)$ accurately approximates our desired estimator $S = \|f(\bz) - \bv\|_2^2 + \|\bv - f(\bq)\|_2^2 + 2 \nu^2 \|\alpha\|_2^2$. 

\begin{lemma}
\label{lem:estlvl}
For each $i\in[H]$, let $\calP$ be the active block in level $i$ and let $\bB_{i}$ be the sketch matrix of $\calP$. Let $S = \|f(\bz) - \bv\|_2^2 + \|\bv - f(\bq)\|_2^2 + 2 \nu^2 \|\alpha\|_2^2$ denote the desired estimator. With high probability, the output $P_i+Q_i$ of $\EstLevel(i)$ satisfies
\begin{align*}
    S- \frac{\eps^2 (6i+2) A}{100H^2} & \le P_i+Q_i\le  S+\frac{\eps^2 (6i+2) A}{100H^2}.
\end{align*}
\end{lemma}
\begin{proof}

Recall that by \Cref{lem:recursion:computation}, $Q_i$ is computed by (1) estimating $\O{L^2}$ fixed inner products using $\bB_i$ and (2) recursively estimating $\|\bw_i \|_2^2 = \|\sum_{j =1}^L \alpha_j \cdot f(\bq_j - \bq_L) - f(\bq_{L+1})\|_2^2$ in the next level. Note that by \Cref{lem:recursion:bounded:norm}, we have $\|\bw_i\|_2^2 \leq \|f(\bz)\|_2^2 + \O{A}$, so in particular $\|\bw_i\|_2^2$ can also be computed with $\frac{6i\eps^2 A}{100H^2}$ additive error.

Suppose $\EstLevel(i-1)$ outputs $Q_i$ such that $\|\bv-f(\bq)\|_2^2 + \nu^2 \|\alpha\|_2^2 - \frac{6i\eps^2 A}{100H^2}\le Q_i\le\|\bv-f(\bq)\|_2^2  + \nu^2 \|\alpha\|_2^2 + \frac{6i\eps^2 A}{100H^2}$. 
Then it follows by \Cref{lem:est:Pi} that with high probability, the output $P_i+Q_i$ of $\EstLevel(i)$ satisfies
\begin{align*}
    S - \frac{\eps^2(6i+1)A}{100H^2}&\le P_i+Q_i \le S + \frac{\eps^2(6i+1)A}{100H^2}. 
\end{align*}
Since the argument above shows the inductive hypothesis for a fixed $i \in [H]$, it suffices for us to verify that the statement holds for the base case $i =1$. 
In particular, $Q_1=\EmbedEst(\bv, \bq)$ is computed using sketch $\bB_1$. 
By the correctness of $\bB_1$, we have 
\[\|\bv-f(\bq_0+\bq_1)\|_2^2 +\nu^2 \|\alpha\|_2^2 - \frac{\eps^2 A}{100H^2}\le Q_1\le \|\bv-f(\bq_0+\bq_1)\|_2^2 + \nu^2 \|\alpha\|_2^2 + \frac{\eps^2 A}{100H^2},\]
with high probability. Thus, by a union bound over all $i \in [H]$ it follows that the desired bound holds for all $i \in [H]$ by induction.
\end{proof}

We now argue that $\EstLevel(i)$ returns an estimate within additive error $\frac{\O{\eps^2 A} \cdot i}{100H^2}$ of the $F_p$ moment $\|\bz - \bq\|_p^p$.
\begin{invariant}
\label{invar:output:acc:p}
For any level $i\in[H]$, let $\calP$ be the active block in level $i$ and let $\bB_{i}$ be the sketch matrix of $\calP$. 
Let $\bq$ be the part of the query in $\calP$ and let $\bz$ be previous part. 
Then with high probability 
\[\|\bz-\bq\|_p^p - \frac{(6i+1)\eps A}{100H}\le\EstLevel(i,\bB_i\bz)\le \|\bz-\bq\|_p^p + \frac{(6i+1)\eps A}{100 H}.\]
\end{invariant}
\begin{proof}
Suppose that the current iterate is not updated  $\bv$ for current suffix $\bq$. Then, by construction we have that 
\[\mathcal{A}_i(\bz - \bq) - \frac{6i\eps A}{100 H}\leq P_i + Q_i \leq \mathcal{A}_i(\bz - \bq) + \frac{6i\eps A}{100 H}.\]
By the correctness of $\mathcal{A}_i$, we have that 
\[\|\bz - \bq\|_p^p - \frac{\eps A}{100 H}\leq \mathcal{A}_i(\bz - \bq) \leq \|\bz - \bq\|_p^p + \frac{\eps A}{100 H}.\]
Therefore, 
\[\|\bz- \bq\|_p^p  - \frac{(6i+1)\eps A}{100 H}\leq P_i + Q_i \leq \|\bz - \bq\|_p^p + \frac{(6i+1)\eps A}{100 H},\]
so the invariant is satisfied in this case. Alternatively, suppose that the iterate $\bv$ is updated for the current suffix $\bq$. In this case, we update $\bv$ by solving the regularized regression problem (see \Cref{alg:get:iterate}) and return the estimate $A_i$, as desired.
\end{proof}

Next, we show that the estimator in block $C_{i,j}$ of level $i \in [H]$ is inaccurate at most $\O{\frac{1}{\eta^2} \log n}$ times throughout the stream, where $\eta = \frac{\eps}{100 H}$ is the total allowable error in each of the $H$ levels of the recursion. 
Recall that the estimator at level $i$ is given by $\|f(\bz) - \bv \|_2^2 + \| \bv - f(\bq)\|_2^2 + 2 \nu^2 \|\alpha\|_2^2$, where $\bz$ represents the previous frequency vector, $\bv$ is the projection of $f(\bz)$ onto the list of previous queries on which the estimator had erred, and $\bq$ contains the arriving updates. Since our estimator simulates the robust $L_2$ algorithm by embedding into $L_2^2$, progress of our algorithm will also be measured in the $L_2$ space.

\begin{lemma}[Bounded list size]
\label{lem:list:size:ltwo}
Let $L_i$ be the number of times that the iterate $\bv$ in a particular block at level $i\in[H]$ is updated and let $\eta = \Theta\left(\frac{\eps}{H}\right)$ be the error in level $i$. 
With high probability, $L_i\le\O{\frac{1}{\eta^2}\log n}$ for all $i\in[H]$ and for all times in the stream, for $m\le\poly(n)$.
\end{lemma}
\begin{proof}

For any fixed level $i\in[H]$ and for any fixed time in the stream, let $\calP_i$ and $\calP_{i+1}$ be the active blocks at levels $i$ and $i+1$ respectively. 
Let $\bz$ be the frequency vector corresponding to all the updates prior to $\calP_{i+1}$, and let $\by $ be the vector $f(\bz)$ with $L$ zeros in the additional dimensions, so that $\by \in \mathbb{R}^{N + L}$. Similarly, let $\bu$ be the vector $\bu = f(\bq)+\nu\cdot\be_{N+|\calL_i|+1}$, and $\bw$ is precisely the projection of $\by$ onto the span of $\bu_i- \bu_1$ for $i \in [|\calL_i|]$, as defined in \Cref{alg:get:iterate}. Conditioned on the correctness guarantees of $\EstLevel$ in \Cref{alg:est:level:smallp}, we have that 
\[\|f(\bz) - \bv\|_2^2 + \|\bv - f(\bq)\|_2^2 + 2 \nu^2 \|\alpha\|_2^2 - \eta A \leq \EstLevel(i, \bB_{i} \bz)\leq \|f(\bz) - \bv\|_2^2 + \|\bv - f(\bq)\|_2^2 + 2 \nu^2 \|\alpha\|_2^2 + \eta A,\]
where $\bv$ is our true iterate. 
Recall that by our assumption, $\|f(\bz)\|_2^2 = \|\by\|_2^2 =  \Theta(A)$ and $\|f(\bz) - f(\bq)\|_2^2 = \Theta(A)$, so we also have $\|\by - \bu\|_2^2 = \Theta(A)$.
In particular, by our choice of $\nu = \frac{\eps^2 A}{H^2 \cdot \sqrt{L}}$ and by \Cref{lem:coefficient:bound}, we must have $\|\by- \bw\|_2^2 = \Theta(A)$ and $\|\bw - \bu\|_2^2 \leq \O{A}$. 
Therefore, it follows that
\[(1-\eta') \cdot (\|\by-\bw\|_2^2+\|\bw-\bu\|_2^2)\le\EstLevel(i,\bB_i\bz) \le(1+\eta')\cdot(\|\by-\bw\|_2^2+\|\bw-\bu\|_2^2),\]
for any choice of $\eta'=\Theta(\eta)$. Let $\bu'$ be the projection of $\bu$ onto the span of $\bu_2- \bu_1,\ldots, \bu_{L} - \bu_{1}$. 
By Pythagorean theorem, we have 
\[(1-\eta') \cdot (\|\by-\bw\|_2^2+\|\bw-\bu'\|_2^2+\|\bu'-\bu\|_2^2)\le\EstLevel(i,\bB_{i}\bz) \le(1+\eta')\cdot(\|\by-\bw\|_2^2+\|\bw-\bu'\|_2^2+\|\bu'-\bu\|_2^2).\]
Suppose that we do not have
\[(1-\eta)^2\cdot\|f(\bz)-f(\bq)\|_2^2\le\EstLevel(i,\bB_{i}\bz)\le(1+\eta)^2\cdot\|f(\bz)-f(\bq)\|_2^2.\]
Then, by the choice of $\nu$ above, it directly follows that we also do not have 
\[(1-\eta)^2\cdot\|\by-\bu\|_2^2\le\EstLevel(i,\bB_{i}\bz)\le(1+\eta)^2\cdot\|\by-\bu\|_2^2.\]
Now, we consider the cases:
\begin{enumerate}
    \item $\|\by-\bw\|_2^2+\|\bw-\bu'\|_2^2+\|\bu'-\bu\|_2^2\le 16\|\by-\bu\|_2^2$ or
    \item $\|\by-\bw\|_2^2+\|\bw-\bu'\|_2^2+\|\bu'-\bu\|_2^2>16\|\by-\bu\|_2^2$.
\end{enumerate}
First, suppose $\|\by-\bw\|_2^2+\|\bw-\bu'\|_2^2+\|\bu'-\bu\|_2^2\le 16\|\by-\bu\|_2^2$. 
Then, 
\begin{align*}
    \left\lvert\|\by-\bw\|_2^2+\|\bw-\bu'\|_2^2+\|\bu'-\bu\|_2^2-\|\by-\bu\|_2^2\right\rvert &>2\eta\cdot\|\by-\bu\|_2^2-\eta'\cdot(\|\by-\bw\|_2^2 \\ &+\|\bw-\bu'\|_2^2+\|\bu'-\bu\|_2^2).
\end{align*}
By Pythagorean theorem, we have
\begin{align*}
\|\by-\bu\|_2^2&=\|(\by-\bw)+(\bw-\bu')+(\bu'-\bu)\|_2^2\\
&=\|\by-\bw+\bu'-\bu\|_2^2+\|\bw-\bu'\|_2^2\\
&=\|\by-\bw\|_2^2+\|\bu'-\bu\|_2^2+2\langle \by-\bw,\bu'-\bu\rangle+\|\bw-\bu'\|_2^2.
\end{align*}
Hence if we have a bad estimator, then 
\[\lvert\langle \by-\bw,\bu'-\bu\rangle\rvert>\eta\cdot\|\by-\bu\|_2^2-\frac{\eta'}{2}\cdot(\|\by-\bw\|_2^2+\|\bw-\bu'\|_2^2+\|\bu'-\bu\|_2^2).\]
Since $\|\by-\bw\|_2^2+\|\bw-\bu'\|_2^2+\|\bu'-\bu\|_2^2\le 16(\|\by-\bu\|_2^2)$ by assumption, for $\eta'\le\frac{\eta}{16}$ we have
\[\lvert\langle \by-\bw,\bu'-\bu \rangle\rvert>\frac{\eta}{2}\cdot\|\by-\bu\|_2^2.\]
Now, let $\by'$ be the projection of $\by$ to the span of the new subspace augmented with $\bu$. 
By Pythagorean theorem,
\begin{align*}
\|\by-\by'\|_2^2&=\|\by-\bw\|_2^2-\frac{1}{\|\bu-\bu'\|_2^2}\cdot\langle \by-\bw,\bu-\bu'\rangle^2\\
&\le\|\by-\bw\|_2^2-\frac{\eta^2}{4}\cdot\frac{\|\by-\bu\|_2^4}{\|\bu-\bu'\|_2^2}\\
&\le\|\by-\bw\|_2^2-\frac{\eta^2}{16}\cdot\frac{\|\by-\bw\|_2^2\cdot\|\bu-\bu'\|_2^2}{\|\bu-\bu'\|_2^2}\\
&=\left(1-\frac{\eta^2}{16}\right)\cdot\|\by-\bw\|_2^2. 
\end{align*}
In particular, note that $\bu-\bu'\neq\mathbf{0}^n$, since $\bu$ cannot lie in the span of $\calL_{i+1}$ while still corresponding to a bad estimate. 

In the other case, we have $\|\by-\bw\|_2^2+\|\bw-\bu'\|_2^2+\|\bu'-\bu\|_2^2>16(\|\by-\bu\|_2^2)$. 
Then
\[\left\lvert\|\by-\bw\|_2^2+\|\bw-\bu'\|_2^2+\|\bu'-\bu\|_2^2-\|\by-\bu\|_2^2\right\rvert>\frac{15}{16}\cdot(\|\by-\bw\|_2^2+\|\bw-\bu'\|_2^2+\|\bu'-\bu\|_2^2).\]
Recall that
\[\|\by-\bu\|_2^2=\|\by-\bw\|_2^2+\|\bu'-\bu\|_2^2+2\langle \by-\bw,\bu'-\bu\rangle+\|\bw-\bu'\|_2^2.\]
So, if the estimator is not accurate, then we must have 
\[\lvert\langle \by-\bw,\bu'-\bu\rangle\rvert>\frac{1}{4}\cdot(\|\by-\bw\|_2^2+\|\bw-\bu'\|_2^2+\|\bu'-\bu\|_2^2).\]
Next, let $\by'$ be the projection of $\by$ to the span of the new subspace augmented with $\bu$. 
By Pythagorean theorem, we have that
\begin{align*}
\|\by-\by'\|_2^2&=\|\by-\bw\|_2^2-\frac{1}{\|\bu-\bu'\|_2^2}\cdot\langle \by-\bw,\bu-\bu'\rangle^2\\
&\le\|\by-\bw\|_2^2-\frac{1}{16}\cdot\frac{(\|\by-\bw\|_2^2+\|\bw-\bu'\|_2^2+\|\bu'-\bu\|_2^2)^2}{\|\bu-\bu'\|_2^2}\\
&\le\|\by-\bw\|_2^2-\frac{1}{16}\cdot\frac{\|\by-\bw\|_2^2\cdot\|\bu-\bu'\|_2^2}{\|\bu-\bu'\|_2^2}\\
&=\frac{15}{16}\cdot\|\by-\bw\|_2^2. 
\end{align*}

Note that the argument above assumes that we compute the new projection $\by'$ exactly. However, 
\Cref{alg:est:level:smallp} computes distances to $\by'$ up to additive error $\Theta\left(\frac{\eps^2 A}{H^2}\right)$. Since the total progress is $\O{\eta^2 A}$ for $\eta = \frac{\eps}{H}$ and the error incurred in $\by'$ is $\Theta\left(\frac{\eps^2 A}{H^2}\right)$, we have shown that $\|\by-\by'\|_2^2\le(1-\O{\eta^2})\cdot\|\by-\bw\|_2^2$. 

By assumption, we have $\|\by\|_2^2 = \|f(\bz)\|_2^2= \Theta(A) \leq \poly(n)$ initially, so the total number of vectors added to the list is at most $\O{\frac{1}{\eta^2}\log n}$, as desired.
\end{proof}

\subsection{Achieving Multiplicative Error}
\label{sec:mult:error}
In this section, we show how to convert an algorithm that estimates a function $F(\cdot)$ up to additive $\eps \cdot (1+\eps)^i$ error to an algorithm that achieves a multiplicative $(1+\eps)$-approximation guarantee. 
This allows us to complete our adversarially robust $F_p$ moment estimation algorithm on insertion-deletion streams.  
Let $\calA_1, \ldots, \calA_{K}$ denote $K = \O{\frac{1}{\eps}\log n}$ independent copies of a robust additive-error algorithm, where $\calA_i$ achieves additive $\eps \cdot (1+\eps)^i$ error when the frequency vector $\bx$ satisfies $F(\bx) \leq 10 \cdot (1+\eps)^i$.

\begin{figure*}[!htb]
\begin{mdframed}\label{alg:mult:error}
\textbf{Algorithm}:
\begin{enumerate}
\item For each stream element $u_t = (a_t, \Delta_t)$, update $\calA_1,\ldots, \calA_K$ with $u_t$. Let $\bx$ denote the frequency vector for the stream, after this update.
\item For $j = K,\ldots, 1:$
\begin{enumerate}
\item 
Compute $S_j \gets \calA_j(\bx)$. 
\item 
If $S_j \geq (1 + \frac{\eps}{4})(1+\eps)^{j-1}$, return $S_j$.
\item 
Otherwise, move to algorithm $\calA_{j-1}$.
\end{enumerate}
\end{enumerate}
\end{mdframed}
\end{figure*}

In the next lemma, we show that this algorithm  returns a $(1\pm \eps)$ approximation to $F(\bx)$ and is robust to $m = \poly(n)$ adaptive updates.

\begin{lemma}\label{lem:mult:error}
Given any $\eps \in (0,1)$, there exists an adversarially robust insertion-deletion streaming algorithm on a stream of length $m = \poly(n)$ that with high probability, outputs a multiplicative $(1\pm \O{\eps})$ approximation to $F(\bx)$. 
\end{lemma}
\begin{proof}
We first argue that the algorithm above is robust. Note that at each scale $j \in [K]$, the additive-error algorithm $\calA_j$ is guaranteed to return an estimate $S_j$ such that $|F(\bx) - S_j| \leq \eps \cdot (1+\eps)^j$, with high probability when the input $\bx$ satisfies $F(\bx) \leq 10 \cdot (1+\eps)^j$. Since the algorithm starts with $j = K$ and iterates backwards, this condition is satisfied automatically, so the algorithm $\calA_j$ will satisfy this additive guarantee. Furthermore, if $S_j$ is returned by $\calA_j$, it must be that $S_j \geq (1 + \frac{\eps}{4})(1+\eps)^{j-1}$ and $S_{j+1} < (1+\frac{\eps}{4})(1+\eps)^{j}$, so
\[
F(\bx) \ge (1 + \frac{\eps}{4})(1 + \eps)^{j - 1} - \O{\eps} \cdot (1 + \eps)^j
\]
and
\[
F(\bx) <(1 + \frac{\eps}{4}) (1 + \eps)^j + \O{\eps} \cdot (1 + \eps)^{j + 1} \;.
\]
From the above, we see that $S_j \in (1 \pm \O{\eps}) \cdot F(\bx)$. Thus, if each $\calA_i$ has additive error $\eps \cdot (1+\eps)^j$ when $F(\bx) \leq 10 \cdot (1+\eps)^j$, then the output is a $(1+\O{\eps})$-multiplicative approximation. 
So, it suffices to show that each $\calA_i$ has small additive error $\eps \cdot (1+\eps)^i$ for inputs $F(\bx) \leq 10 \cdot (1+\eps)^j$ with probability $1-\frac{1}{n^2}$, even when the stream is chosen adaptively. To see this, we argue by contradiction: suppose there exists a  adaptive strategy $\calS$ of length $\poly(n)$ for the adversary, such that for some $j \in [K]$, $\calA_j$ has error probability $\geq \frac{1}{n^2}$. We design an adaptive strategy to break the additive-error algorithm $\calA_j$ as follows: the adversary can simulate the execution of the algorithm for copies $\{\calA_1',\ldots, \calA_{j-1}', \calA_j, \calA_{j+1}',\ldots \calA_L'\}$, where $\calA_j$ is the embedded target copy. Then, the adversary can simply execute the adaptive strategy $\calS$ on $\{\calA_1',\ldots, \calA_j,\ldots, \calA_L' \}$, and by assumption $\calS$ will ensure that $\calA_j$ fails to output a $(1\pm\eps)$ approximation with probability at least $\frac{1}{n^2}$. However, $\calA_j$ was assumed to be robust to $m = \poly(n)$ adaptive inputs, so this results in a contradiction.
\end{proof}

We now show correctness for our adversarially robust $F_p$ estimation algorithm for $p \in [0,2)$, which achieves a $(1 + \eps)$-approximation to $F_p$ using polylogarithmic space.

\thmlpsmall*
\begin{proof}
We begin by noting that a randomized adversary can be modeled as a distribution over sequences of possible inputs. 
Since the probability that such an adversary succeeds is at most the maximum success probability over any fixed input sequence, it suffices without loss of generality to analyze the case of a deterministic adversary.  

We now upper bound the number of computation paths by upper bounding the number of interactions with the corrector. Recall that at each step, \MaintainList queries each $F_p$ sketch $\bB_i$ whether it should update the current iterate $\bv$ with query $\bq$, which would effectively add the embedding of the query $\bq$ to a list $\calL_{i}$ and solve the regularized regression problem in the new subspace. As a result, the transcript from $\bB_i$ can be viewed as a sequence over the alphabet $\{\bot, \top\}$, where the output $\bot$ from the sketching matrix indicates that no action should be taken (i.e. the estimator is sufficiently accurate for $\bq$) and $\top$ indicates that the iterate should be updated by adding the embedding of $\bq$ to the list $\calL_{i}$. If $\top$ occurs at most $L$ times in a single block of the level below, and there are $B$ blocks in the level below, it follows that the total number of times that the corrector will be used to flag an inaccurate query is $B \cdot L$ overall. 

Note that if the iterate is updated, this causes the linear sketch at level $i-1$ to maintain $\bB_{i-1} \bv'$, where $\bv'$ is the iterate obtained by solving the regularized regression problem. Additionally, the linear sketch $\bB_{i-1}$ releases the estimate $A_{i}$, but no other information about $\bB_i$ is revealed. 
Additionally, at the start of each of the $B$ blocks, the algorithm outputs an estimate of the $F_p$ moment, encoded using $C_1\log n$ bits for some constant $C_1 > 0$. Thus, the total number of possible output streams from the sketch matrix is bounded by  
\[\binom{m}{L\cdot B}\cdot (C_1\log n)^{B\cdot L + B}.\]
Since the adversary is deterministic, each possible input stream is determined by one of these output streams.  
By setting the total failure probability to  
\[\delta \le \frac{1}{(nm)^3}\cdot\left(\binom{m}{L\cdot B}\cdot (C_1\log n)^{B \cdot L + B}\right)^{-1},\]
a union bound over all possible adversarial input streams guarantees that the sketch matrix at level $i$ is correct with probability at least $1 - \frac{1}{(nm)^3}$.  
For $m = \poly(n)$, it follows that
\[\log\frac{1}{\delta} = \O{(B \cdot L + B)\log n}.\]

Next, we count the number of adaptive interactions with the estimator, which consists of $P_i$ and $Q_i$. To start, we use a bounded computation paths argument to track the adaptive interactions due to $P_i$. First, we note that the output of each sketch is rounded to $\O{\log n}$ bits. Furthermore, the number of adaptive interactions with $\bB_i$ is precisely the number distinct vectors which are sketched by $B_i$. Recall that the iterate $\bv$ is updated at most $L$ times due to some hard query $\bq_i$, and to estimate $P_i = \|f(\bz) - \bv\|_2^2$ we must compute $\O{L^2}$ inner products. Thus, it follows that the number of computation paths is $\binom{m}{L^2} \cdot \left(2^{C_2 \log n}\right)^{L^2}$ for some constant $C_2 > 0$. As before, it suffices to set $\log \frac{1}{\delta} = \O{(B \cdot L^2 + B) \log n}$. In particular,observe that each adaptive update to $P_i$ corresponds to a different value for the iterate $\bv$, and other queries to $P_i$ are to the same value of the iterate $\bv$ and are thus correct conditioned on the correctness of the estimate for $\bv$. For robustness of the estimate $Q_i$, we observe that by \Cref{lem:recursion:computation}, $Q_i$ consists of a term $\|\sum_{i = 1}^L \alpha_i \cdot f(\bq_i - \bq_{L}) - f(\bq_{L+1})\|_2^2$ which is the input to the recursion, and $\O{L^2}$ inner products between $f(\bz)$, $f(\bq_i)$, and $f(\bq_i - \bq_{L})$. Since the addition of each of the $f(\bq_i)$ to the list corresponds to a single update to the iterate $\bv$, and the $\bq_i$ are fixed, it follows that the inner products discussed in \Cref{lem:recursion:computation} can be computed robustly given that $\bB_i$ is robust to $\O{L^2}$ adaptive interactions.

This establishes robustness, after which correctness follows directly from \Cref{invar:output:acc:p}. By setting the error parameter $\beta = \frac{\eps^2}{\kappa^2 \cdot L^2 \cdot H^2 }$ as in \Cref{lem:est:Pi}, we ensure that the total allowable error for each level of the algorithm is $\frac{\eps A}{100 H} = \eta \cdot \Theta(A)$.

At this point, recall that the algorithm $\calA$ achieves additive $\eps \cdot A$ error when the frequency vector $\bx = \bz - \bq$ has $F_p$ moment $\|\bx\|_p^p\leq 10A$, for a fixed parameter $A \leq \poly(n)$, and is robust to $m = \poly(n)$ adaptive updates. By applying the framework in \Cref{alg:mult:error}, we obtain an algorithm $\calA'$ which returns a multiplicative $(1\pm \eps)$-approximation to $\|\bx\|_p^p$ at all times in the stream, with high probability. 

\paragraph{Space complexity.}
We now analyze the total space of the algorithm. 
For a fixed tree and level $i \in [H]$, we maintain $B$ sketch matrices, one for each active block at level $i+1$. Recall $\eta = \frac{\eps}{100H}$ is the total error for each level. Each sketching algorithm must achieve accuracy $\left(1+\beta\right)$ and failure probability $\delta$, where  
\[\log\left(\frac{1}{\delta}\right) = \O{(B \cdot L^2 +B)\log n}.\]
By \Cref{lem:list:size:ltwo}, we have $L \le \frac{C_3}{\eta^2}\log n$ for some constant $C_3 > 0$, with high probability, conditioned on the correctness of the subroutines.  
Since we require $B^H \ge m$ and $B \ge L$, it suffices to set  
\[B = \O{\frac{1}{\eta^2}\log n},\qquad H = \O{\log n},\]
so that $\log \frac{1}{\delta} = \O{\frac{1}{\eta^6} \log^4 n}$. 

By \Cref{thm:obliv:small:fp}, the space required to store the sketch $\bB_i$ is $\O{\frac{1}{\beta^2} \log^2 n \log \frac{1}{\delta}}$. 
There are $B = \O{\frac{1}{\eta^2} \log n}$ sketching matrices in each level across $H$ levels of recursion; moreover, the multiplicative-error framework in \Cref{alg:mult:error} requires storing $\O{\log n}$ copies of this additive-error algorithm. 
Therefore, the final algorithm uses $\poly\left(\frac{1}{\eps},\log n\right)$ bits of memory.
\end{proof}

\ignore{
\begin{remark}[Improved space bound via differential privacy]
    As an alternative to the bounded computation paths argument, one can apply the differential privacy-based robustification framework of \cite{HassidimKMMS20} to protect sketches $\bB_i$ from adaptive interactions with the adversary in the argument above. By Theorem 3.4 in their work, each sketch $\bB_i$ uses

    \[\O{\frac{1}{\eta^2} \log n \log \left(\frac{1}{\delta}\right) \cdot \sqrt{(B \cdot L + B) \cdot \log \left(\frac{1}{\delta}\right)} \cdot \log \left(\frac{m}{\eta \delta} \right)}\]

    bits of space, where $m \leq \poly n$ and $\log \left(\frac{1}{\delta}\right) = \O{\frac{1}{\eta^4}\log^3 n}$, and $\eta = \frac{\eps}{H}$. Substituting into the bound above, we obtain \elena{fill.}
\end{remark}}

\section{Embedding Applications}

\subsection{Embedding into \texorpdfstring{$L_p$}{Lp}}
In this section, we give a number of additional applications of our adversarially robust $L_p$ estimation algorithm for $p\in[0,2)$. 
Unlike the applications in \Cref{sec:huber-loss} which embed into $L_2^2$, the applications in this section embed into $L_p$ for $p\in[0,2)$. 
To that end, recall that a result by \cite{AndoniKR18} guarantees that any norm $\calX = (\mathbb{R}^n, \|\cdot\|_{\calX})$ satisfying triangle inequality and sub-multiplicativity of norms admits a classical sketch and embeds linearly into a measure space $L_p(\mu)$ that solves the following problem:

\begin{definition}[Gap-norm problem]
Given a parameter $D>1$, a vector $\bx\in\mathbb{R}^n$ and a norm $\|\cdot\|_{\calX}$, distinguish between the input cases $\|\bx\|_{\calX}\le 1$, versus $\|\bx\|_{\calX}\ge D$. 
\end{definition}

\begin{theorem}
\label{thm:akr}
\cite{AndoniKR18}
Suppose $p\in\left(\frac{2}{3},1\right)$ and $\calX$ is a norm that admits an oblivious linear sketch for the $D$-gap norm problem using $s$ space, then $\calX$ admits a linear embedding $T: \calX\to L_p(\mu)$ with distortion $\O{sD}$. 
\end{theorem}

However, the proof of \Cref{thm:akr} is existential and not necessarily algorithmically computable in general. 
Thus to efficiently evaluate this embedding, we utilize an advice model where the streaming algorithm is granted read-only access to a static advice string, i.e., the embedding matrix $\bA$, that depends solely on the metric space $\calX$ and dimension $n$.
In particular, we assume that stream updates arrive as discrete increments in the standard computational basis $\be_1, \ldots, \be_n$ of $\mathbb{R}^n$ and that the metric space is well-conditioned with respect to this basis; otherwise, the norm of a simple unit stream update could exceed the representable word size. 
Formally:
\begin{assumption}[Basis condition number]
\label{assumption:basis:condition}
Let $R = \max_{j=1}^n \|\bu_j\|_{\calX}$ denote the maximum norm of any standard computational basis vector $\be_1,\ldots,\be_n$, and let $r = \min_{x \neq 0} \frac{\|\bx\|_{\calX}}{\|\bx\|_p}$ denote the minimum scaling factor relative to the standard $L_p$ norm. 
We assume that the geometric condition number defined by $\kappa = \frac{R}{r}$ satisfies $\kappa \le \poly(n)$. 
\end{assumption}

Given that the distortion of the embedding of \Cref{thm:akr} is $\O{sD}$, we cannot subsequently embed $L_p$ into $L_2^2$, as this could cause a blow-up of $(sD)^{\O{H}}$ across $H$ different recursive levels. 
Instead, we observe that the embedding of \Cref{thm:akr} is crucially linear, so we can directly map the input stream to a virtual stream over $\mathbb{R}^M$, where $M$ is the dimension of the embedding. 
Our robust $L_p$ estimation algorithm is then executed over the virtual stream. 
However, this requires some care about the magnitude of $M$, as well as the magnitude of the resulting entries, which we now observe. 

\begin{lemma}
\label{lem:embedding:properties}
Let $\calX = (\mathbb{R}^n, \|\cdot\|_{\calX})$ be a finite-dimensional normed space. 
Let $\bA\in\mathbb{R}^{M\times n}$ be the static linear embedding into $L_p^M$ with distortion $D' = \O{sD}$. 
Then the dimension is upper bounded by $M=\O{n\log n}$ and there exists a basis for $\calX$ such that every entry of $\bA$ can be represented exactly using $\O{\log n+\log\frac{sD}{\eps}}$ bits.
\end{lemma}
\begin{proof}
We first show an upper bound on the target dimension $M$. 
By \Cref{thm:akr}, there exists a linear map $T: \calX \to L_p(\mu)$ for some $p \in (2/3, 1)$ with distortion $\O{sD}$. 
By the subspace discretization theorem of \cite{schechtman1987more,talagrand1990embedding,schechtman2001embedding}, any $n$-dimensional subspace of $L_p$ linearly embeds into $L_p^M$ with distortion $1+\eps$ using $M = \O{\frac{n}{\eps^2}\log n}$ coordinates, e.g., using Lewis Weight sampling. 
Hence, we have $M=\O{n\log n}$ by fixing the distortion to be a constant, i.e., $\gamma=\Theta(1)$. 

To upper bound the magnitude of the entries, recall that by Auerbach's Lemma, there exists a basis $\be_1, \ldots, \be_n$ for $\calX$ such that $\|\be_j\|_{\calX} = 1$ and for any $\bx = \sum x_j \be_j$, we can upper bound the coefficients by $\max_j |x_j| \le \|\bx\|_{\calX}$. 
Let $\bA^* \in \mathbb{R}^{M \times n}$ be the exact continuous embedding matrix under this basis. 
Then the maximum magnitude of any entry is upper bounded by its $L_p$ column norm: 
\[ \max_i |A^*_{i,j}| = \left(\max_i |A^*_{i,j}|^p\right)^{1/p} \le \left(\sum_{i=1}^M |A^*_{i,j}|^p\right)^{1/p} = \|\bA^*\be_j\|_p \le D' \|\be_j\|_{\calX} = D'=\O{sD}. \]

To lower bound the magnitude of the nonzero entries, consider a discrete matrix $\bA$ formed by rounding every entry of $\bA^*$ to the nearest multiple of $\nu = \frac{\eps}{M^{1/p} n}$. 
For any stream vector $\bx \in \mathbb{R}^n$, the additive error in the $L_p^p$ quasi-norm induced by this perturbed matrix $\bA$ is at most:
\[\|\bA\bx - \bA^*\bx\|_p^p = \sum_{i=1}^M \left|\sum_{j=1}^n (A_{i,j} - A^*_{i,j}) x_j\right|^p \le \sum_{i=1}^M \left(\sum_{j=1}^n \nu |x_j|\right)^p \le \sum_{i=1}^M (\nu n \max_j |x_j|)^p \le M \nu^p n^p \|\bx\|_{\calX}^p = \eps^p \|\bx\|_{\calX}^p.\]
Because $L_p$ is a quasi-norm for $p\le1$, we have the subadditivity inequality $\|\bu+\bv\|_p^p \le \|\bu\|_p^p + \|\bv\|_p^p$. 
Thus, the additive error translates to the metric bounds as:
\[ \|\bA\bx\|_p^p \le \|\bA^*\bx\|_p^p + \|\bA\bx - \bA^*\bx\|_p^p \le \|\bA^*\bx\|_p^p + \eps^p\|\bx\|_{\calX}^p \]
and symmetrically, $\|\bA\bx\|_p^p \ge \|\bA^*\bx\|_p^p - \eps^p\|\bx\|_{\calX}^p$.
This implies that discretizing to the grid of granularity $\nu$ changes the $L_p^p$ distortion bound by a negligible additive factor of $\eps^p$. 
Moreover, since the entries reside in $[-D', D']$ on a grid of step $\nu$, the total number of representable values is $\frac{2D'}{\nu}=\O{\frac{D'M^{1/p}n}{\eps}}$. 
Hence, every entry $A_{i,j}$ can be represented using $\O{\log n+\log\frac{sD}{\eps}}$ bits, which gives the claimed bit complexity. 
\end{proof}

Thus, we have the following corollary to \Cref{thm:akr}:
\begin{corollary}
\label{cor:sketch:to:alg}
Suppose $\calX$ is a norm that admits an oblivious linear sketch for distinguishing distance $1$ vs $D$ using $s$ space. 
Then there exists an adversarially robust streaming algorithm for $\calX$ that, given oracle access to the corresponding embedding matrix, uses $S\cdot\polylog(n)$ space and achieves a multiplicative $\O{sD}$ approximation.  
\end{corollary}

Putting together \Cref{cor:sketch:to:alg} with known equivalences between turnstile algorithms and linear sketches, we have the following:
\thmequivalence*
\begin{proof}
Note that the forward direction follows immediately by \Cref{cor:sketch:to:alg}, while the reverse direction holds for polynomially-long streams from known results, e.g., \cite{LiNW14,AiHLW16,HosseiniLY19, JiangLY26}. 
\end{proof}

Finally, we give a number of specific applications of \Cref{cor:sketch:to:alg} where the linear transformation $\bA$ is explicit.

\paragraph{Robust graph total variation.}
In graph signal processing, the graph total variation is a standard tool for quantifying the smoothness of a signal $\bx \in \mathbb{R}^V$ defined on a graph $G=(V,E)$. 
Fixing an arbitrary orientation of the edges, the corresponding incidence matrix $\bB \in \{-1,0,1\}^{|E| \times |V|}$ maps vertex signals to edge-wise differences. 
The graph total variation norm can be written as $\|\bx\|_{\text{TV}} = \|\bB\bx\|_1$, i.e., the $L_1$ norm of edge differences. 
The matrix $\bB$ thus provides an isometric linear embedding from $(\mathbb{R}^V, \|\cdot\|_{\text{TV}})$ into $(\mathbb{R}^E, \|\cdot\|_1)$. 
Then applying \Cref{cor:sketch:to:alg}, we have:
\begin{theorem}[Robust Graph Total Variation]
\label{thm:graph:tv}
There exists an adversarially robust turnstile streaming algorithm that computes a $(1+\eps)$-approximation to the graph total variation seminorm $\|\bx\|_{\text{TV}(G)} = \sum_{(u,v)\in E} |x_u - x_v|$, using $\poly\left(\frac{1}{\eps},\log|V|\right)$ bits of space. 
\end{theorem}

\paragraph{Robust trend filtering ($k$-th order total variation).}
Consider a temporal signal $\bx \in \mathbb{R}^n$, which can be viewed as a sequence of observations over time, e.g., measurements or prices. 
A natural way to quantify the variability of such a signal is through its discrete differences, where the first-order differences capture successive changes $x_{i+1} - x_i$ and higher-order differences capture more refined notions of variation, such as curvature for $k=2$ and higher-order fluctuations. 
Formally, the $k$-th order difference operator is represented by a matrix $\bD^{(k)}$, obtained by iteratively applying a total of $k$ times the first-order operator $\bD^{(1)} \in \mathbb{R}^{(n-1)\times n}$, which has $1$ on the main diagonal and $-1$ on the superdiagonal, i.e., the diagonal that is directly above and to the right of the main diagonal. 
In this way, the $k$-th order total variation norm is defined as $\|\bx\|_{TV^{(k)}} = \|\bD^{(k)}\bx\|_1$, so that it is the sum of the absolute values of these higher-order differences. 
Intuitively, this quantity measures how irregular the signal is at scale $k$. 
For example, $k=1$ penalizes sharp jumps, while $k=2$ penalizes changes in slope, encouraging piecewise linear structure. 
This notion forms the basis of trend filtering, where one seeks signals that balance fidelity to data with low higher-order variation.
Note that since $\bD^{(k)}$ can be explicitly computed, the $k$-th order total variation norm embeds explicitly and isometrically into $L_1$. 
Thus by applying \Cref{cor:sketch:to:alg}, we have:
\begin{theorem}[Robust Trend Filtering]
\label{thm:trend:filter}
There exists an adversarially robust turnstile streaming algorithm that computes a $(1+\eps)$-approximation to the $k$-th order total variation norm $\|\bx\|_{TV^{(k)}} = \|\bD^{(k)}\bx\|_1$, using $\poly\left(\frac{1}{\eps},\log n\right)$ bits of space.
\end{theorem}

\paragraph{Robust Earth Mover distance (EMD).}
Earth Mover distance (EMD) is a common metric for comparing two distributions over a discrete grid $[\Delta]^d$, so that $n = \Delta^d$. 
Intuitively, EMD quantifies the minimal ``work'' needed to transform one distribution into another, where work is the amount of mass moved times the distance the mass is moved. 
We interpret $\|\bx\|_{\EMD}$ as the Earth Mover Distance between the signed measure $\bx$ and the zero vector $\mathbf{0}$, i.e., $\|\bx\|_{\EMD} := \EMD(\bx, \mathbf{0})$. 
Equivalently, for two distributions $\mu,\nu$, we have $\EMD(\mu,\nu) = \|\mu - \nu\|_{\EMD}$. 
Perhaps the most well-used embeddings for EMD are based on randomly shifted quadtrees and give $\O{d\log\Delta}$ distortion, but provide a \emph{``for each''} guarantee. 
That is, for a fixed signal $\bx$, the embedding preserves distances in expectation over its random choices. 
However, this randomized embedding is insufficient for our purposes in the adversarially robust streaming setting, as an adaptive adversary can exploit the particular random shift to create worst-case distortions.

Instead, consider the following construction of a fully deterministic embedding. 
Let $\bA_{\bs} : \mathbb{R}^n \to \mathbb{R}^K$ denote the linear embedding induced by a single quadtree shift $\bs$, with $K = \O{n}$. 
We define the deterministic, concatenated embedding
\[\bA\bx = \frac{1}{n} \bigoplus_{s \in [\Delta]^d} \bA_{\bs}\bx,\]
which aggregates the contributions from all possible shifts. 
By linearity of expectation, the $L_1$ norm of $\bA\bx$ exactly matches the expected value over all shifts:
\[\|\bA\bx\|_1 = \frac{1}{n} \sum_{\bs \in [\Delta]^d} \|\bA_{\bs}\bx\|_1 = \EEx{\bs}{\|\bA_{\bs}\bx\|_1}.\]
This gives a deterministic, universal \emph{``for all''} guarantee, so that for all $\bx\in\mathbb{R}^n$, we have
\[\|\bx\|_{\EMD} \le \|\bA\bx\|_1 \le \O{d \log \Delta}\|\bx\|_{\EMD}.\]
Although the embedding dimension grows from $K=\O{n}$ to $M=nK=\O{n^2}$, the resulting $L_1$ estimation algorithm only uses space polylogarithmic in the dimension. 
Hence by applying \Cref{cor:sketch:to:alg}, we have:
\thmemd*

\paragraph{Robust $k$-median clustering.}
Next, we consider the $k$-median clustering problem over the grid $[\Delta]^d$, where a dataset is represented as a multiset of points, or equivalently, a mass vector $\mu \in \mathbb{R}^{[\Delta]^d}$. 
The $k$-median objective seeks a set of $k$ centers $\calC\subseteq [\Delta]^d$ minimizing the total $L_2$ distance from each point to its nearest center. 
Equivalently, for an assignment vector $\nu_{\calC}$ that routes each unit of mass in $\mu$ to its closest center in $\calC$, the $k$-median cost can be expressed as an Earth Mover Distance:
\[\Cost_k(\mu) = \min_{\calC: |\calC|=k} \EMD(\mu, \nu_{\calC}).\]
Thus, to approximate this quantity in the streaming setting, previous works often utilize the standard embedding of $\EMD$ into $L_1$ via randomly shifted quadtrees. 
Fix a shift $s \in \mathbb{Z}^d$ and suppose without loss of generality that $\Delta=2^\ell$ for some integer $\ell>0$.  
Then for each scale $t \in \{0,1,\ldots,\ell\}$, define a grid $\mathcal{G}_{s,t}$ with side length $2^t$. 
For a mass vector $\mu$, let $G_{s,t}\mu$ denote the vector recording the total mass in each cell of $\mathcal{G}_{s,t}$. 
Then the embedding is defined by concatenation across scales:
\[G_s \mu = (G_{s,0}\mu)\circ (2 \cdot G_{s,1}\mu)\circ \cdots \circ (2^t \cdot G_{s,t}\mu)\circ \cdots \circ (2^\ell \cdot G_{s,\ell}\mu).\]
The key intuition is that the contribution of a unit of mass to the $L_1$ norm $\|G_s(\mu - \nu)\|_1$ is determined by the finest scale at which the mass assigned in $\mu$ and $\nu$ fall into different cells. 
Because coarser grids are weighted more heavily, this construction ensures that the $L_1$ norm faithfully captures transportation cost. 
In particular, this embedding provides an upper bound on $\EMD(\mu,\nu)$ up to a factor of $\O{d \log \Delta}$.

As in the EMD setting, to obtain adversarial robustness we derandomize the construction by concatenating over all shifts, yielding a deterministic embedding $\bA$ with a \emph{``For All''} guarantee. 
Applying \Cref{cor:sketch:to:alg} to the embedded vector allows us to approximate $\|\bA(\mu - \nu_{\calC})\|_1$ for all candidate center sets $\calC$ simultaneously in sublinear space. 

Finally, we recall that there exists a collection $\calC$ of candidate sets of centers and assigned masses such that $|\calC|\le(\Delta m)^{\O{kd}}$, one of which can be used to find a constant-factor approximation to the optimal $k$-median clustering cost~\cite{BackursIRW16,Cohen-AddadWZ23}. 
For our purposes, it suffices to enumerate over this candidate set, though we remark that there are more efficient top-down search processes to find a good candidate set~\cite{BackursIRW16,Cohen-AddadWZ23}. 
Therefore, we have:
\thmkmedian*

\subsection{Applications to Entropy Estimation}
\label{sec:entropy}
In this section, we leverage our adversarially robust $F_p$ estimation algorithms to obtain a black-box procedure for estimating entropy. 
Recall that given a frequency vector $\bv\in\mathbb{R}^n$, the Shannon entropy is defined as $H(\bv)=-\sum_{i=1}^n v_i\log v_i$. 
We first recall the following relationship between the Shannon entropy and its exponentiation:
\begin{observation}
\label{obs:entropy:addmult}
An $\eps$-additive approximation to the Shannon entropy $H(v)$ is a $(1+\eps)$-multiplicative approximation to $h(v):=2^{H(v)}$ and vice versa. 
\end{observation}
In light of this observation, it suffices to focus on computing a $(1+\eps)$-multiplicative approximation to $h(v)=2^{H(v)}$. 
To that end, we recall to following structural property that facilitates entropy approximation. 
\begin{lemma}[Section 3.3 in~\cite{HarveyNO08}]
\label{lem:entropy:reduction}
Let $k=\log\frac{1}{\eps}+\log\log m$ and $\eps'=\frac{\eps}{12(k+1)^3\log m}$. 
There exists an efficiently computable collection $\{y_0,\ldots,y_k\}$ with each $y_i\in(0,2)$, together with a deterministic procedure that, given $(1+\eps')$-approximations to $F_{y_i}(v)$ for all $i$, outputs a $(1+\eps)$-approximation to $h(v)=2^{H(v)}$. 
\end{lemma}
The points $\{y_0,\ldots,y_k\}$ from \Cref{lem:entropy:reduction} can be constructed explicitly as described in~\cite{HarveyNO08}. 
Let $\ell=\frac{1}{2(k+1)\log m}$ and define 
\[f(z)=\frac{(k^2\ell)z-\ell(k^2+1)}{2k^2+1}.\]
Then each $y_i$ is given by $y_i=1+f\left(\cos\left(\frac{i\pi}{k}\right)\right)$, which allows the entire set to be computed in linear time. 
Finally, a $(1+\eps)$-multiplicative estimate of $h(v)=2^{H(v)}$ is obtained from $2^{P(0)}$, where $P(x)$ is the degree-$k$ polynomial interpolating the points $\{(y_i, F_{y_i}(v))\}_{i=0}^k$. 
Crucially, since $y_i\in(0,2)$ for each $i$, we can apply our $F_p$ estimation algorithm to get an additive $\eps$-approximation to the Shannon entropy:
\thmentropy*

\subsection{Embedding into \texorpdfstring{$L_2^2$}{L22}}
\label{sec:huber-loss}
In this section, we show that our algorithm for $L_p$ estimation, $p \in [0,2)$ can be extended to a broader class of functions. 
In particular, consider the task of estimating $\norm{\bx}_g := \sum_i g(x_i)$ for some specific function $g$. 
We show that if $g$ has the form $g(t) = f(t^2)$ for some Bernstein function $f$ (\Cref{def:bernstein-function}), then we can design an adversarially robust streaming algorithm which computes a $(1 \pm \eps)$-approximation to $\norm{\bx}_g$ at each time in the adaptive stream, using $\poly\left(\frac{1}{\eps},\log n\right)$ bits of space. 
We note that this definition captures many interesting choices of the function $g$, which are relevant for various applications.

\begin{itemize}
\item 
(Pseudo–Huber Loss): $g_{\tau}(x) =  \tau(\sqrt{1 + (x / \tau)^2} - 1)$ 
\item 
(Cauchy/Lorentzian Loss): $g_{\tau}(x) = \log (1 + x^2 / \tau)$
\item 
(Generalized Charbonnier): $g_{\tau} = (1 + x^2 / \tau)^\beta - 1$ for $0 < \beta \le 1$.
\item 
(Welsch/Leclerc Loss): $g_{\tau}(x) = 1 - e^{-x^2 / \tau}$ 
\item 
(Geman–McClure loss):  $ \displaystyle g_{\tau} =\frac{x^2}{x^2 + \tau}$
\end{itemize}
 
We first recall the following classical result in~\cite{schoenberg1938metric}.
\begin{definition}[Negative Type (NT) Function]
\label{def:nt}
A symmetric continuous function $\psi : \mathbb{R}^d \to \mathbb{R}$ is NT if $\psi(0) = 0$, $\psi(x) \geq 0$, and for any $n \in \mathbb{N}$, $\{x_i\} \subset \mathbb{R}^d$ and $\{\alpha_i\} \subset \mathbb{R}$ such that $\sum \alpha_i = 0$,
\[\sum_{i,j} \alpha_i\alpha_j\psi(x_i - x_j) \leq 0.\]
\end{definition}

\begin{theorem}[\cite{schoenberg1938metric}]
\label{thm:embedding}
Let $\psi: \mathbb{R}^d \to \mathbb{R}$ be a continuous function satisfying the conditions of \Cref{def:nt}. Then the following are equivalent:
\begin{enumerate}
\item 
$\psi$ is of negative type.
\item 
There exists a Hilbert space $H$ and a mapping $\Phi$ such that $\psi(x - y) = \|\Phi(x) - \Phi(y)\|_H^2$.
\item 
$f_t(x) = e^{-t\psi(x)}$ is Positive Definite for all $t > 0$.
\end{enumerate}
\end{theorem}

\begin{definition}[Bernstein function, e.g., \cite{schilling2012bernstein}]
\label{def:bernstein-function}
A function $f:(0,\infty)\to[0,\infty)$ is called a \emph{Bernstein function} if $f\in C^\infty(0,\infty)$ and for all $n\in\mathbb{N}_0$ and all $x>0$,
\[(-1)^n f^{(n+1)}(x) \ge 0.\]
where $f^{(n)}$ denotes the $n$-th derivative of the function $f$. 
\end{definition}

\begin{theorem}[Theorem 3.1, \cite{KUHN201913}]
\label{thm:bernstein-negative}
Let $f: [0, \infty) \to [0, \infty)$ and $f(0) = 0$. Then the following two statements are equivalent:
\begin{enumerate}
\item 
$f$ is a Bernstein function.
\item 
For $\xi \in \mathbb{R}^k$, $f(|\xi|^2)$ is a continuous negative type function for all $k \ge 1$. 
\end{enumerate}
\end{theorem}

By applying~\Cref{thm:embedding} and~\Cref{thm:bernstein-negative} coordinate-wise and then taking the sum over all coordinates, we have the following theorem. 
Without loss of generality, we can assume that $\Phi(0) = 0$; otherwise, we can simply replace the embedding $\Phi(\bx)$ by $\Phi(\bx) - \Phi(0)$ to satisfy this condition.

\begin{theorem}  
\label{thm:embedding_L2}
Suppose $f: \mathbb{R}^n \to \mathbb{R}$ satisfies $f(x) = \sum_{i=1}^n g(x_i)$, where $g(x) = h(x^2)$ for some Bernstein function $h$ and $h(0) = 0$.
There exists an isometric embedding $\Phi$ from $\mathbb{R}^n$ into a Hilbert space $H$, such that
\[f(\bu-\bv)=\|\Phi(\bu)-\Phi(\bv)\|_2^2,\qquad \|\Phi(0)\|_2^2=0.\]
\end{theorem}

To utilize our algorithm for $L_p$ estimation for $p \in[0, 2)$, we also need a (non-robust) algorithm for estimating $\norm{\bx}_g$. 
In the work of~\cite{GribelyukLWYZ26}, the authors give the following lemma, which is based on the zero-one law of the form of the function $g$ in~\cite{BCWY16}.

\begin{lemma}[\cite{GribelyukLWYZ26, BCWY16}]
\label{lem:nra}
Given a function $g: \mathbb{Z}_{n \ge 0} \to \mathbb{R}$ where $g(x) = f(x^2 )$ for some Bernstein function $f$, there exists a one-pass (non-robust) turnstile streaming algorithm that $(1 \pm \eps)$-estimates the value of $\norm{\bx}_g = \sum_{i = 1}^n g(x_i)$ with high probability in space $\poly\left(\frac{1}{\eps},\log n\right)$.
\end{lemma}

Finally, by combining \Cref{thm:embedding_L2} and \Cref{lem:nra}, we obtain adversarially robust turnstile streaming algorithms that compute a $(1+\eps)$-approximation for any $g$-sum $\|\bx\|_g$, where $g(x) = h(x^2)$ for some Bernstein function $h$ and $h(0) = 0$.

\begin{theorem}
Given a function $g: \mathbb{Z}_{n \ge 0} \to \mathbb{R}$ where $g(x) = f(x^2 )$ for some Bernstein function $f$ and $f(0) = 0$, there exists an adversarially robust insertion-deletion streaming algorithm on a stream of length $m$ that with high probability, outputs a $(1 \pm \eps)$-approximation to the value of $\norm{\bx}_g = \sum_{i = 1}^n g(x_i)$. 
For $m = \poly(n)$, the algorithm uses $\poly\left(\frac{1}{\eps},\log n\right)$ bits of space.
\end{theorem}

\section*{Acknowledgements}
Elena Gribelyuk and Huacheng Yu are supported in part by an NSF CAREER award CCF-2339942. 
Honghao Lin was supported in part by a Simons Investigator Award and a CMU Paul and James Wang Sercomm Presidential Graduate Fellowship. 
David P. Woodruff is supported in part by Office of Naval Research award number N000142112647 and a Simons Investigator Award. 
Samson Zhou is supported in part by NSF CCF-2335411 and gratefully acknowledges funding provided by the Oak Ridge Associated Universities (ORAU) Ralph E. Powe Junior Faculty Enhancement Award.

\def\shortbib{0}
\bibliographystyle{alpha}
\bibliography{references}

\end{document}